\documentclass[sn-mathphys]{sn-jnl}% Math and Physical Sciences Reference Style
\jyear{2022}%

\usepackage{amssymb,amsmath,graphicx}

\usepackage{amsfonts,amsthm} % Math packages
\usepackage{bm,bbm}
\usepackage{physics}
\usepackage{mathtools}
\usepackage{xcolor}
\usepackage{hyperref}
\usepackage{tensor}
\usepackage{subcaption} % TO ORGANIZE THE FIGURES
\theoremstyle{thmstyleone}%
\theoremstyle{thmstyletwo}%

\theoremstyle{thmstylethree}%

\newcommand{\vc}{\vcentcolon =}

\newcommand{\bone}{\mathbbm{1}}	% math bold one

\newcommand{\sC}{\mathbb{C}}
\newcommand{\us}{\underline{\sigma}}
\newcommand{\ur}{\underline{\rho}}
\newcommand{\uF}{\underline{F}}

\DeclareMathOperator{\diag}{diag}

\def\fC{\mathbb{C}}
\def\fR{\mathbb{R}}

\def\ra{\rangle}
\def\la{\langle}

\newtheorem{thm}{Theorem}[section]
\newtheorem{prop}[thm]{Proposition}
\newtheorem{lem}[thm]{Lemma}

\theoremstyle{definition}

\numberwithin{equation}{section}

\newcounter{mnotecount}[section]

\begin{document}

\title[Monotonicity of the quantum 2-Wasserstein distance]{\center{Monotonicity of the quantum} 
\\  \center{2-Wasserstein distance}}


\author[1]{Rafa\l\ Bistro\'n}\email{rafal.bistron@student.uj.edu.pl}

\author[1]{Micha\l\ Eckstein}\email{michal.eckstein@uj.edu.pl}

\author[1,2]{Karol {\.Z}yczkowski}\email{karol.zyczkowski@uj.edu.pl}

\affil[1]{Institute of Theoretical Physics, 
Jagiellonian University, ul. {\L}ojasiewicza 11, 30--348 Krak\'ow, Poland}

\affil[2]{Center for Theoretical Physics, Polish Academy of Sciences, Al. Lotnik\'{o}w 32/46, 02-668 Warszawa, Poland}


\abstract{
We study a quantum analogue of the 2-Wasserstein distance as a measure of proximity on the set $\Omega_N$ of density matrices of dimension $N$. We show that such (semi-)distances do not induce Riemannian metrics on the tangent bundle of $\Omega_N$ and are typically not unitary invariant. Nevertheless, we prove that for $N=2$ dimensional Hilbert space the quantum 2-Wasserstein distance (unique up to rescaling) is monotonous with respect to any single-qubit quantum operation and the solution of the quantum transport problem is essentially unique. Furthermore, for any $N \geq 3$ and the quantum cost matrix proportional to a projector we demonstrate the monotonicity under arbitrary mixed unitary channels.
Finally, we provide numerical evidence which allows us to conjecture that the unitary invariant quantum 2-Wasserstein semi-distance is monotonous with respect to all CPTP maps in any dimension $N$.
}



%\date{March 31, 2022}

\maketitle

\section{Introduction}
\label{sec:intro}

The distances on the space of probability vectors are a primary tool to compare and analyse various statistical distributions. On the mathematical side, they lie at the heart of information geometry \cite{Amari1,Amari2}, which unveils a powerful interplay between Riemannian geometry and statistics. In particular, with some distances on the space of probability vectors one can associate a Riemannian metric on the suitable tangent bundle. In the set of all such distances the one generated by the Fischer--Rao metric is distinguished as the unique continuous distance monotone under classical stochastic maps (Cencov theorem \cite{cencov}).

In the modern field of quantum information \cite{NielsenChuang,BZ17} classical probability vectors are replaced by density matrices. The quantum analogue of a stochastic map is a completely positive trace preserving (CPTP) linear map, $\Phi: \Omega_N \to \Omega_M$, where $\Omega_N \vc \{\rho \in \sC^{N \times N} \: \vert \: \rho=\rho^*, \ \rho\ge 0, \ \Tr \rho=1\}$ is the set of density matrices of order $N$. A semi-distance $D_{\text{mon}}$ on the space of quantum states is called \emph{monotone} if the inequality
	\begin{align}\label{mono}
	 D_{\text{mon}} ( \Phi(\rho^A), \Phi(\rho^B) ) \leq  D_{\text{mon}} ( \rho^A, \rho^B )
	\end{align}
holds for all states $\rho^A, \rho^B \in \Omega_N$, for any dimension $N$, and all  CPTP maps $\Phi$. 
	
%The set of all monotone distances on the space of quantum states, inducing a Riemannian metric on the tangent bundle of $\Omega_N$, can be characterised, via the Morozova--Cencov--Petz theorem \cite{Morozova,Petz1}, through a single-parameter operator-monotone function. %, $f(t)$. / - MCP theorem odnosi się do samych metryk
The set of all monotone Riemannian metrics on the tangent bundle of $\Omega_N$, can be characterised, via the Morozova--Cencov--Petz theorem \cite{Morozova,Petz1}, through a single-parameter operator-monotone function. %, $f(t)$.
On the other hand, there exist important monotone distances, which do not induce a Riemannian metric --- most notably, the trace distance $D_{\text{Tr}}(\rho^A,\rho^B) = \Tr \vert \rho^A - \rho^B \vert$, with $\vert X \vert = \sqrt{XX^\dagger}$.

The property of monotonicity is crucial for applications in quantum information processing. It implies that the distinguishability, quantified by a given distance function, cannot increase under an application of any quantum map.

\medskip

%It seems that monotonicity of Projective quantum transport cost, discussed below, is essential to proper functioning of recently proposed quantum Wasserstein Generative Adversary Networks\cite{qWGAN} on a present-day noisy hardware.

First attempts to generalise the classical Monge transport distance to the quantum setting with the use of Husimi function were done in \cite{ZS98,ZS01}. More recently, a more general approach of Kantorovich \cite{Kan48} and Wasserstein \cite{Vil09} was extended to the quantum case an analysed from the mathematical \cite{AF17,CGGT17,BGJ19,Fri20,Duv20}, physical \cite{GMT16,GP18,CGP20,CM20,Ikeda20} and information-theoretical  \cite{CGNT17,PMTL20,PT21,KdPMLL21} perspective.  
%,transport1,transport2} 
%\ME{Karolu, sprawdz prosze liste cytowan.}. 
Among these, a recent proposal \cite{transport1,transport2} is based on the optimization over the set of bipartite quantum states with fixed marginals of the expectation value of a certain `cost observable' (see also \cite{qWGAN,YZYY18,Reira18}). Concretely, for any classical distance $d$ on the set of $N$ points one can define \cite{transport2} an associated \emph{quantum cost operator}
\begin{align}\label{costE}
  C^Q_E = \sum_{j>i=1}^{N} E_{ij}  \vert \psi^{-}_{ij} \rangle \langle \psi^{-}_{ij} \vert,
\end{align} 
where $E_{ij} = d(x_i,x_j)$ is the distance between points $x_i, x_j$, while $\vert \psi^{-}_{ij} \rangle \langle \psi^{-}_{ij} \vert$ is the projector on the antisymmetric subspace spanned by two base vectors $\vert i \rangle$, $\vert j \rangle$, i.e. $\vert \psi^{-}_{ij} \rangle = \tfrac{1}{\sqrt{2}} (\vert i \rangle \vert j \rangle - \vert j \rangle \vert i \rangle)$. For any two density matrices $\rho^A, \rho^B \in \Omega_N$ one introduces the set $\Gamma^Q(\rho^A, \rho^B)$ of bi-partite coupling matrices $\rho^{AB}$, such that both partial traces are fixed, $\Tr_A \rho^{AB} = \rho^B$ and $\Tr_B \rho^{AB} = \rho^A$. Minimising the Hilbert--Schmidt scalar product of the cost matrix and $\rho^{AB}$ over all possible coupling matrices we arrive at the \emph{optimal quantum transport cost}, 
\begin{align}\label{TQ}
T^{Q}_E(\rho^A, \rho^B) \vc
 \min_{\rho^{AB}\in \Gamma^Q(\rho^A,\rho^B)}  \Tr\ \big( C^Q_E \rho^{AB} \big).
\end{align}
It was shown in \cite{transport2}  that $T^{Q}_E$ is a weak metric (i.e. a semi-metric bounded from below by a genuine metric) on $\Omega_N$ for any $N$ and any quantum cost matrix $C^Q_E$. Following the classical analogy, one can thus define \cite{transport1} the corresponding \emph{quantum 2-Wasserstein semi-distance}
\begin{align}\label{WE}
W_E(\rho^A, \rho^B) \vc
 \sqrt{ \min_{\rho^{AB}\in \Gamma^Q(\rho^A,\rho^B)}  \Tr\ \Big( (C^Q_E)^2 \rho^{AB} \Big)}.
\end{align}

 On the 2-point set all classical distance matrices $E$ are equivalent up to rescaling. In higher dimensions a distinguished role is played by the simplex geometry, for which $E_{ij} = 1- \delta_{ij}$. In this case, the corresponding quantum cost matrix \eqref{costE} forms a projector onto the antisymmetric subspace of the Hilbert space $\sC^{N \times N}$. For sake of brevity we shall denote such a quantum cost matrix by
\begin{align}\label{CQ}
C^Q \vc \tfrac{1}{2} (\bone_N - S), \quad \text{ where } \quad S (\vert x \rangle \vert y \rangle) = \vert y \rangle \vert x \rangle
\end{align}
and, since $(C^Q)^2 = C^Q$, write
\begin{align}\label{W}
W(\rho^A, \rho^B) \vc \sqrt{T^{Q}(\rho^A, \rho^B)}, \; \text{ with } \; T^{Q}(\rho^A, \rho^B) \vc \min_{\rho^{AB}}  \ \Tr \big( C^Q \rho^{AB} \big).
\end{align}

For single-qubit states, $N=2$, it was shown \cite{transport1,transport2} that $W$ enjoys the triangle inequality and hence is a genuine distance on the Bloch ball $\Omega_2$. Furthermore, the numerical simulations strongly suggest that the triangle inequality actually holds for the quantum 2-Wasserstein semi-distance \eqref{W} with the cost matrix \eqref{CQ} in any dimension $N$. However, in the problem of monotonicity the triangle inequality does not play any role, hence we shall --- for simplicity --- focus on the general semi-distances $T^Q_E$ and the specific case of $T^Q$ determined by the projection cost matrix \eqref{CQ}.

The quantum optimal transport problem \eqref{TQ} admits a dual formulation \cite{qWGAN,transport2}. The optimization takes place over the set of pairs $\sigma^A$, $\sigma^B \in H_N$ of $N \times N$ Hermitian matrices, satisfying a certain algebraic constraint determined by the cost matrix $C_E$. Concretely, let 
\begin{align}\label{SN}
\Sigma_N \vc \{ \sigma^A, \sigma^B \in H_N \; \vert \; F \vc C_E - \sigma^A \otimes \bone_N - \bone_N \otimes \sigma^B \geq 0\},
\end{align}
then
\begin{align}\label{dual}
T^{Q}_E(\rho^A, \rho^B) = \sup_{(\sigma^A, \sigma^B) \in \Sigma_N} \; \Tr(\sigma^A\rho^A + \sigma^B\rho^B).
\end{align}
If neither of the states $\rho^A, \rho^B$ is pure, then the supremum is achieved \cite{transport2}.

The quantum optimal transport cost \eqref{TQ} can be efficiently computed using semidefinite programming \cite{transport2}. Moreover it seems to be a valuable tool for quantum
machine learning.
%SDP jest niezalezne od ML
In particular, the quantum 2-Wasserstein semi-distance \eqref{W} has been shown to play a key role in the quantum Generative Adversarial Network scheme \cite{qWGAN}. From the viewpoint of applications in machine learning, the monotonicity of a distance is a highly desirable property, which ensures robustness against the noise of the learning algorithms.

\medskip

The primary purpose of this work is to study the monotonicity with respect to the quantum channels of the optimal quantum transport cost $T^Q_E$ defined above. We start in Section \ref{sec:unitary} by
considering some general properties of \eqref{TQ}
and show
that for $N\ge 3$ and
a general classical cost matrix $E$ 
the corresponding quantum
transport cost is not unitarily invariant.
% the associated distance is not unitarily invariant. 
%a one can only expect the monotonicity to hold in the projective case.
Then,
we discuss a quantity associated with the projective transport distance --- the `SWAP-fidelity' --- and 
demonstrate,  in Section \ref{sec:Riemman}, 
that $T^Q_E$ does \emph{not} induce a Riemannian metric, hence the standard Morozova--Cencov--Petz theorem does not apply. Section \ref{sec:gen} includes a proof of monotonicity of $T^Q$ under general mixed unitary channels in any dimension $N$. Then, in Section \ref{sec:mono2} we focus on $N=2$ and provide a complete proof of monotonicity of $T^Q$ for arbitrary single-qubits channels. Finally, in Section \ref{sec:unique}, we show that for any two mixed non-isospectral qubits there exists a unique optimal coupling $\rho^{AB}$ yielding the minimum in \eqref{TQ}. Furthermore, we provide explicit formulae for the optimal coupling, and the optimal dual observables $\sigma^A$, $\sigma^B$, for isospectral or commuting qubits.

The article is supplemented by two appendices. Appendix \ref{sec:proofs} contains the proofs of technical results, while in Appendix \ref{sec:num} we provide substantial numerical evidence for the monotonicity of $T^Q$ under all qutrit and ququart channels. Based on these observations, we are tempted to conjecture that the unitary invariant quantum 2-Wasserstein semi-distance $W$ is actually monotone for all CPTP maps.

\section{No unitary invariance for general quantum 2-Wasserstein semi-distances}
\label{sec:unitary}

In the single-qubit case, $N=2$, the only (up to a trivial multiplicative factor) quantum cost matrix \eqref{costE} is a projector in the antisymmetric subspace \eqref{CQ}. But already in the qutrit case, $N=3$, there exists different cost matrices $C^Q_E$, for instance the one induced by the classical Euclidean distance on the line, $E_{12} = E_{23} = 1$, $E_{13} = 2$ (see the Supplemental Material in \cite{transport1}).

For the projection matrix $C^Q$ the optimal quantum transport cost \eqref{TQ} is invariant under unitary channels:
\begin{align*}
T^Q(U \rho^A U^\dagger, U \rho^B U^\dagger) = T^Q(\rho^A, \rho^B), \quad \text{ for any } \quad U \in \mathrm{U}(N).
\end{align*}
This stems from the fact that $(U \otimes U) C^Q (U^\dagger \otimes U^\dagger) = C^Q$, for any unitary matrix $U \in \mathrm{U}(N)$.

The latter property, however, does not hold for a general quantum cost matrix \eqref{costE}. This implies that one cannot expect the monotonicity to hold in full generality for $T^Q_E$, even under unitary channels. Indeed, let us take two qutrit states
\begin{align*}
\rho^A = \tfrac{1}{2} \big( \vert 1 \rangle + \vert 2 \rangle \big) \big( \langle 1 \vert + \langle 2 \vert \big), \qquad \rho^B = \tfrac{1}{5} \big( \vert 1 \rangle + 2 \vert 2 \rangle \big) \big( \langle 1 \vert + 2 \langle 2 \vert \big)
\end{align*}
and the quantum cost matrix induced by the line geometry,
\begin{align*}
C^Q_E = \vert \psi^{-}_{12} \rangle \langle \psi^{-}_{12} \vert + \vert \psi^{-}_{23} \rangle \langle \psi^{-}_{23} \vert + 2 \vert \psi^{-}_{13} \rangle \langle \psi^{-}_{13} \vert.
\end{align*}
Consider now a unitary channel, which interchanges the states  $\vert 2 \ra$ and $\vert 3\ra$,
\begin{equation*}
    U = \begin{pmatrix}
    1 & 0 & 0 \\
    0 & 0 & -1 \\
    0 & 1 & 0 \\
    \end{pmatrix}
\end{equation*}
and yields
\begin{align*}
\eta^A & = U \rho^A U^\dagger = \tfrac{1}{2} \big( \vert 1 \rangle + \vert 3 \rangle \big) \big( \langle 1 \vert + \langle 3 \vert \big), \\
\eta^B & = U \rho^B U^\dagger = \tfrac{1}{5} \big( \vert 1 \rangle + 2 \vert 3 \rangle \big) \big( \langle 1 \vert + 2 \langle 3 \vert \big).
\end{align*}
Since all of the involved states are pure, there is only one coupling matrix for both pairs \cite[Lemma A.3]{transport2}, $\rho^{AB} = \rho^A \otimes \rho^B$ and $\eta^{AB} = \eta^A \otimes \eta^B$. We thus have
\begin{align*}
T^Q_E (U \rho^A U^\dagger,U \rho^B U^\dagger) = \Tr C^Q_E \eta^{AB} = \tfrac{1}{10} > 
T^Q_E (\rho^A, \rho^B) = \Tr C^Q_E \rho^{AB} = \tfrac{1}{20}.
\end{align*}

Let us note that the lack of unitary invariance of a distance on the space of quantum states can be a desirable property in certain applications. In particular, the quantum Wasserstein distance of order 1 proposed in \cite{PMTL20}, which is not unitarily invariant, offers improved efficiency in quantum learning algorithms \cite{KdPMLL21}.

\section{The lack of Riemannian structure}
\label{sec:Riemman}

Let us now focus on the case of the projective cost matrix $C^Q$. The corresponding transport cost \eqref{W} can be associated with \emph{SWAP-fidelity} \cite{transport1},
    \begin{equation}\label{fS}
        F_S(\rho^A, \rho^B) \vc \max_{\rho^{AB} \in \Gamma^Q(\rho^A,\rho^B)} \big( \Tr \, S \rho^{AB} \big) = 1 - 2 T^Q(\rho^A,\rho^B).
    \end{equation}
    
Clearly, the monotonicity of $T^Q$ under CPTP maps is equivalent to the reverse monotonicity of the SWAP-fidelity, $F_S(\Phi(\rho^A), \Phi(\rho^B)) \geq F_S(\rho^A, \rho^B)$. The SWAP-fidelity shares many properties with the standard Uhlmann--Jozsa quantum fidelity \cite{Uh76,Jo94} and equals to the latter if either of the states $\rho^A,\rho^B$ is pure \cite{transport1}.

The quantum fidelity $F$ is indeed reverse monotone under all CPTP maps \cite{Jo94}. Consequently, any distance on the space of quantum state determined by a strictly decreasing function $h: [0,1] \to \mathbb{R}^+$, $D_{h,F}(\rho^A, \rho^B) =h(F(\rho^A, \rho^B))$, will automatically be monotone. The prominent examples include the root infidelity, $I \vc \sqrt{1-F}$, the Bures distance, $B \vc \sqrt{2(1-\sqrt{F})}$ and the Bures angle, $A \vc \tfrac{2}{\pi} \arccos \sqrt{F}$. As the SWAP-fidelity \eqref{fS} shares many important properties with the standard Uhlmann--Jozsa fidelity, one might expect it to be reverse monotone as well.

All of the above-mentioned fidelity-based distances generate Riemannian metrics on the tangent bundle of $\Omega_N$, for any $N$. More precisely (see e.g. \cite{SIOR17} and \cite{Amari1}), a semi-distance $d: \Omega_N \times \Omega_N \to \mathbb{R}^+$ generates a Riemannian metric $g$ on $T_\rho\Omega_N$ if the expansion
\begin{align}\label{g}
d(\rho,\rho+t v)^2 = g_\rho(v,v) t^2 + o(t^2), \text{ as } t \to 0^+,
\end{align}
holds for any $\rho \in \Omega_N$ and any $v \in T_\rho\Omega_N = \{v \in H_N, \text{ with } \Tr v = 0\}$.

All monotone Riemannian metrics on the set of quantum states, are characterised via the Morozova--Cencov--Petz theorem  \cite{Morozova,Petz1}. The latter gives an explicit formula for the metric $g$ in terms of a single operator monotone function.
However, there are important distances, which do \emph{not} generate a Riemannian metric, but are nevertheless monotone. A classical example is the $l_1$-distance (the `taxicab distance'), which corresponds to the trace distance in the quantum case. It turns out that the quantum 2-Wasserstein distance \eqref{W} shares this feature with the trace distance. 

Let us focus on the case $N=2$ and show the failure of \eqref{g} for the quantum 2-Wasserstein distance \eqref{W}. Because of the unitary invariance of $W$ we can restrict ourselves to the states in the real slice of the Bloch ball, $\Omega_2^\mathbb{R} \subset \Omega_2$,
\begin{align}\label{rho}
\rho(r,\theta) = \tfrac{1}{2} \big[ \bone + (2r-1) ( \sigma_1 \sin \theta  +  \sigma_3  \cos \theta ) \big], \quad \text{ with } r \in [0,1], \theta \in [0,\pi),
\end{align} 
where $\sigma_i$ denote the Pauli matrices. Consequently, $T_{\rho(r,\theta)}\Omega_2^\mathbb{R}$ is a real vector space diffeomorphic to $\mathbb{R}^2$. Furthermore, by unitary invariance, it is sufficient to consider the tangent space at the point $\rho(r,0)$. Basing on the results of \cite{transport2} we can compute the square of the metric derivative (see e.g. \cite{AGS05} for a precise definition) of the quantum 2-Wasserstein distance $W$ on $\Omega_2$ in the direction of a vector $v \in T_{\rho(r,0)}\Omega_2^\mathbb{R}$.

\begin{prop}\label{prop:F}
For any $r \in (0,1)$ and any tangent vector $v = (v_1,v_2) \in T_{\rho(r,0)}\Omega_2^\mathbb{R}$ we have
\begin{align} 
G(r,v) & \vc \lim_{t \to 0^+} \frac{T^Q \big(\rho(r,0),\rho(r + t v_1, t v_2)\big)}{t^2} \notag \\
& = \max_{\phi \in [0,2\pi)}  \frac{\big( 2 v_1 \cos(\phi) - (2r - 1) v_{2}\sin(\phi)  \big)^2}{16 \big( 1 + (2r-1) \cos(\phi) \big)} \label{our_F} 
\end{align}

\end{prop}
\begin{proof}
The proof, which bases on a semi-analytical formula derived in \cite{transport2}, can be found in Appendix \ref{sec:proofs}. 
\end{proof}

If $W$ would generate a Riemannian metric $g$ on $T\Omega_2$, then formula \eqref{g} would imply that $g$ can be recovered from the equality
\begin{align}\label{gF}
G(r,v) = g_{\rho(r,0)}(v,v) = g_{11}(r) v_1^2 + 2 g_{12}(r) v_1 v_2 + g_{22}(r) v_2^2,
\end{align} 
which should be valid for all $r \in (0,1)$ and $v \in \mathbb{R}^2$. The functions $g_{11}$ and $g_{22}$ can be computed explicitly from the analytic formulae for the 2-Wasserstein distance between commuting and isospectral qubits, \eqref{Tcomm} and \eqref{Tiso}, respectively. On the other hand, one can convince oneself with the help of formula \eqref{our_F} that $g_{12}$ depends not only on $r$, but also on the tangent vector $v$ -- see Fig. \ref{fig:Riem}. Hence, formula \eqref{gF} fails and the square of the metric derivative \eqref{our_F} of the 2-Wasserstein distance does \emph{not} induce a Riemannian metric on $T\Omega_2$.

\begin{figure}[h]
    \centering
        \includegraphics[height=1.5in]{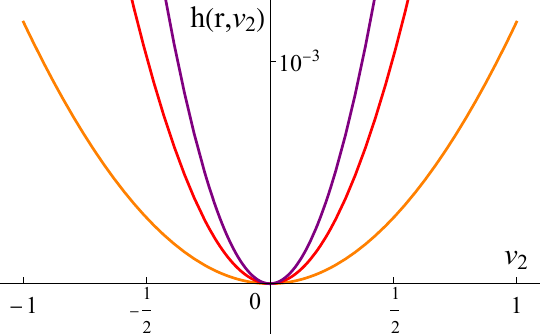}
     \centering
     \caption{\label{fig:Riem}The plot of $h(r,v_2) \vc G \big(r,(1, v_2) \big) - g_{11}(r) - g_{22}(r) v_2^2$ as a function of $v_2$ for $r = 1/3$, $r = 1/4$ and $r = 1/5$ --- colours orange, red and purple, respectively. If Formula \eqref{gF} would be true then, for any $r$, $h(r,v_2)$ would be a linear function of $v_2$, which clearly is not the case.}
\end{figure}

Finally, let us note that the failure of formula \eqref{gF} extends to the general case of $\Omega_N$ and any quantum cost matrix $C^Q_E$. This is because $\Omega_2$ is embedded isometrically in $\Omega_N$ for any $N \geq 3$ and the restriction of any $C^Q_E$ to a suitable four-dimensional subspace of $\mathbb{C}^{N \times N}$ is proportional to $C^Q$ --- see \cite[Prop. 2.4]{transport2}.

\section{Monotonicity under general mixed unitary channels}\label{sec:gen}

Recall that the CPTP maps can be conveniently characterised with the help of Kraus operators. For any channel $\Phi: \Omega_N \to \Omega_M$ there exists a set of $N \times M$ matrices $\{K_i\}_{i = 1}^R$, with $R \leq NM$, called Kraus operators, such that
\begin{equation}
\label{Krauss_def1}
    \sum_i K_i^{\dagger} K_i = \bone_N
\end{equation}
and for any $\rho \in \Omega_N$
\begin{equation}
\label{Krauss_def2}
    \Phi(\rho) = \sum_i K_i \rho K_i^{\dagger}.
\end{equation}
Conversely, quantum channels can be constructed through Kraus operators, since every set of $N \times M$ matrices satisfying \eqref{Krauss_def1} defines a CPTP map via \eqref{Krauss_def2}.

A particular class of quantum channels arises from the statistical mixtures of unitary maps. In such a case, we have $K_i = \sqrt{p_i} U_i$, where $U_i \in U(N)$ and $p_i \geq 0$, $\sum_i p_i = 1$. The general properties of the projective quantum transport cost imply the following result:
\begin{prop}\label{lem_comb1}
Let $\Psi$ be a mixed unitary channel on $\Omega_N$, then for any $\rho^A, \rho^B \in \Omega_N$ we have
\begin{align*}
 T^Q \big( \Psi(\rho^A),\Psi(\rho^B) \big) \leq T^Q(\rho^A, \rho^B).
\end{align*}
\end{prop}
\begin{proof}
This follows directly from the joint convexity of the transport cost \cite[Proposition 2.2]{transport2} and the unitary invariance of $T^Q$.
\end{proof}

The convexity of the optimal quantum transport cost $T^Q_E$ implies a more general result, which holds for any quantum cost matrix \eqref{costE}.

\begin{prop}\label{lem_comb2}
 Let $\Phi_1, \ldots, \Phi_R$ be quantum channels under which the optimal quantum transport cost is monotonous:
 \begin{equation*}
     T_E^Q \big(\Phi_i(\rho^A), \Phi_i(\rho^B) \big) \leq T_E^Q(\rho^A, \rho^B) \quad \text{ for all } \; \rho^A,\rho^B \in \Omega_N \text{ and } i \in \{1,\ldots,R\}.
 \end{equation*}
Then, $T^Q_E$ is monotonous under any convex combination of the channels $\Phi_1, \ldots, \Phi_R$,
 \begin{equation*}
     T_E^Q \big(\Psi(\rho^A), \Psi(\rho^B) \big) \leq T_E^Q(\rho^A, \rho^B) \quad \text{ for all } \; \Psi = \sum_i p_i \Phi_i, \text{ with } p_i\geq 0, \sum_i p_i = 1.
 \end{equation*}
\end{prop}

The last result implies that it is sufficient to study the monotonicity of the optimal quantum transport under the \emph{extremal channels}, i.e. the ones which cannot be decomposed as a convex combination of CPTP maps.

We can also invoke the dual formulation \eqref{dual} to derive the following useful result:
\begin{prop}\label{prop_phi}
Let $\Phi: \Omega_N \to \Omega_M$ be a CPTP map characterised by Kraus operators $\{K_i\}$ and denote by $\Phi^*: H_M \to H_N$ the map dual to $\Phi$, ie. $\Phi^*(\sigma) = \sum_i K_i^\dagger \sigma K_i$. If for any pair of matrices $\big(\sigma^A,\sigma^B\big) \in \Sigma_M$, with $\Sigma_M$ defined in \eqref{SN}, we have $\big(\Phi^*(\sigma^A),\Phi^*(\sigma^B)\big) \in \Sigma_N$, then $T^Q_E$ is monotone with respect to the map $\Phi$, $T^Q_E \big(\Phi(\rho^A),\Phi(\rho^B) \big) \leq T^Q_E(\rho^A, \rho^B)$.
\end{prop}

\begin{proof}
Clearly, if $\sigma$ is a Hermitian matrix, then so is $\Phi^*(\sigma)$. %Assume that neither of the states $\rho^A, \rho^B$ is pure. Then, 
From the dual formulation of the optimal quantum transport problem \eqref{dual}, we deduce
\begin{align*}
T^Q_E \big(\Phi(\rho^A),\Phi(\rho^B) \big) & = \sup_{(\sigma^A, \sigma^B) \in \Sigma_M} \Tr ( \sigma^A \Phi(\rho^A) + \sigma^B \Phi(\rho^B) ) \notag \\
& = \sup_{(\sigma^A, \sigma^B) \in \Sigma_M} \Tr ( \Phi^*(\sigma^A) \rho^A + \Phi^*(\sigma^B) \rho^B ) \notag \\
%& = \Tr ( \us^A \Phi(\rho^A) + \us^B \Phi(\rho^B) ) \notag \\
%& = \Tr ( \Phi^*(\us^A) \rho^A + \Phi^*(\us^B) \rho^B )  \notag \\
& \leq \sup_{(\eta^A, \eta^B) \in \Sigma_N} \Tr \big( \eta^A \rho^A + \eta^B \rho^B \big) = T^Q_E(\rho^A, \rho^B).
\end{align*}
\end{proof}

\section{Monotonicity under arbitrary single-qubit channels}\label{sec:mono2}

Let us now focus on the single qubit case, $N=2$. If we write a state $\rho \in \Omega_2$ using its Bloch vector, $\rho(\vec{r}) = \frac{1}{2} \big( \bone + \vec{r} \cdot \vec{\sigma} \big)$, then the action of any quantum channel $\Phi: \Omega_2 \to \Omega_2$ can be expressed \cite{KR01} in terms of a real 3 $\times$ 3 matrix $M$ and a translation vector $\vec{t} \in \mathbb{R}^3$,
\begin{align*}
\Phi(\rho(\vec{r})) = \rho \big( M \vec{r} + \vec{t}\, \big),
\end{align*}
such that $\Vert M \vec{r} + \vec{t} \, \Vert \leq 1$. Furthermore, the matrix $M$ can be diagonalised with the help of two unitary rotations $U, V \in \mathrm{U}(2)$ as follows,
\begin{align}\label{phid}
\Phi(\rho) = U \big[  \Phi_d \big( V \rho V^\dagger \big) \big]  U^\dagger,
\end{align}
where $\Phi_d$ is a quantum channel with a diagonal matrix $M$.

In \cite{RSW01} it was shown that any such single-qubit map $\Phi_d$, which belongs to the closure of the set of extremal quantum channels, can be realised with two Kraus operators
\begin{equation}\label{k2}
\begin{aligned}
K_1 & = \left[\cos\left(\tfrac{1}{2} v\right)\cos\left(\tfrac{1}{2} u\right)\right] \bone + \left[\sin\left(\tfrac{1}{2} v\right)\sin\left(\tfrac{1}{2} u\right)\right] \sigma_3 \ ,\\
K_2 & = \left[\sin\left(\tfrac{1}{2} v\right)\cos\left(\tfrac{1}{2} u\right)\right] \sigma_1 + i \left[\cos\left(\tfrac{1}{2} v\right)\sin\left(\tfrac{1}{2} u\right)\right] \sigma_2 \ ,
\end{aligned}
\end{equation}
with $u \in [0, 2 \pi)$ and $v \in [0,\pi)$. For such quantum channels we have the following technical result proven in Appendix \ref{sec:proofs}.
\begin{lem}\label{lem_F}
Let $\Phi: \Omega_2 \to \Omega_2$ be a CPTP map determined by Kraus operators \eqref{k2} and assume that $\sigma^A,\sigma^B \in H_2$ are diagonal. If $F = C^Q - \sigma^A \otimes \bone_N - \bone_N \otimes \sigma^B \geq 0$ then also $F^\Phi = C^Q - \Phi^*(\sigma^A) \otimes \bone_N - \bone_N \otimes \Phi^*(\sigma^B) \geq 0$.
\end{lem}

It turns out that the matrices $\sigma^A, \sigma^B$ yielding the maximum in \eqref{dual} actually commute and thus can be simultaneously diagonalised.
\begin{lem}\label{lem_sim_diag}
Let  $\rho^A$, $\rho^B \in \Omega_2$ be mixed quantum states 
of full rank and with different spectra, and let $\us^A$, $\us^B$ be two Hermitian matrices saturating the dual problem \eqref{dual}, i.e.
\begin{equation*}
    T^Q(\rho^A,\rho^B)  = \Tr \big( \us^A \rho^A + \us^B \rho^B \big).
\end{equation*}
Then, $\us^A$ and $\us^B$ commute.
\end{lem}
\begin{proof}
This fact is shown in the course of the proof of Theorem 5.1 in \cite{transport2}.
\end{proof}

We are now in a position to present the main result of this paper.

\begin{thm}\label{thm:qBit_monoton}
The optimal quantum transport cost is monotonous under all CPTP maps $\Phi: \Omega_2 \to \Omega_2$.
\end{thm}

\begin{proof}
Thanks to the unitary invariance of the transport cost $T^Q$ we can restrict ourselves to channels $\Phi_d$ defined by Eq. \eqref{phid}. Furthermore, by Proposition \ref{lem_comb2}, it is sufficient to check the monotonicity for extremal single-qubit channels. Such channels are completely characterised by two Kraus operators \eqref{k2}, with $u \in [0,2\pi)$, $v \in [0,\pi)$ such that either $\sin u \cdot \sin v \neq 0$ or $\vert \cos u \cdot \cos v \vert < 1$ (see Theorem 1.4 in \cite{RSW01}). We can thus use the characterisation \eqref{k2} and apply Lemma \ref{lem_sim_diag}.

Now, assume that the states $\rho^A, \rho^B \in \Omega_2$ are non-isospectral and neither of them is pure. By Lemma \ref{lem_sim_diag} we know that the Hermitian matrices saturating the dual problem, $\us^A, \us^B$, commute. Let then $V \in U(2)$ be the matrix diagonalizing simultaneously $\us^A$ and $\us^B$. Using the unitary invariance of $T^Q$ we obtain
\begin{align*}
T^Q \big( \Phi(\rho^A),\Phi(\rho^B) \big) & = T^Q \Big( V^\dagger \Phi(\rho^A) V, V^\dagger \Phi(\rho^B) V \Big) \\
& = \max_{(\sigma^A, \sigma^B) \in \Sigma_2} \Tr ( \sigma^A V^\dagger \Phi(\rho^A) V + \sigma^B V^\dagger \Phi(\rho^B) V) \\
& = \Tr ( \us^A V^\dagger \Phi(\rho^A) V + \us^B V^\dagger \Phi(\rho^B) V) \\
& = \Tr ( \Phi^*(V \us^A V^\dagger) \rho^A + \Phi^*(V \us^B V^\dagger) \rho^B ).
\end{align*}
Now, Lemma \ref{lem_F} implies that $\big( \Phi^*(V \us^A V^\dagger), \Phi^*(V \us^B V^\dagger) \big) \in \Sigma_2$, hence we conclude that
\begin{align*}
T^Q \big( \Phi(\rho^A),\Phi(\rho^B) \big) & = \Tr ( \Phi^*(V \us^A V^\dagger) \rho^A + \Phi^*(V \us^B V^\dagger) \rho^B ) \\
& \leq \sup_{(\eta^A, \eta^B) \in \Sigma_2} \Tr \big( \sigma^A \rho^A + \sigma^B \rho^B \big) = T^Q(\rho^A,\rho^B ).
\end{align*}
This inequality extends by the continuity of $T^Q$ and $\Phi$ to the limiting cases of isospectral and pure states $\rho^A, \rho^B$.
\end{proof}

%%%%%%%%%%%%
%%%%%%%%%%%%
%%%%%%%%%%%%
%%%%%%%%%%%%
%%%%%%%%%%%%
%%%%%%%%%%%%
%%%%%%%%%%%%
%%%%%%%%%%%%
%%%%%%%%%%%%
%%%%%%%%%%%%
%%%%%%%%%%%%
%%%%%%%%%%%%

\section{Solutions to the single-qubit transport problem}\label{sec:unique}

In this section, we study the solutions to original \eqref{TQ} and dual \eqref{dual} optimal quantum transport problems for single-qubit states.

Let us start with some general remarks, which apply in any dimension $N$. While the original problem always has at least one solution yielding the minimum in \eqref{TQ}, the supremum in the dual problem \eqref{dual} is attained only if both of the states of the states $\rho^A$, $\rho^B$ are positive definite (see \cite[Theorem 3.2]{transport2}). On the other hand, if at least one of the states is pure, then there exists only one coupling, $\rho^{AB} = \rho^A \otimes \rho^B$ --- see \cite[Lemma A.3]{transport2}. Note also that the solution of the dual problem \eqref{dual} is always given up to a shift: $(\sigma^A, \sigma^B) \to (\sigma^A - c \bone_2, \sigma^B + c \bone_2)$, where $c \in \fR$ is an arbitrary constant. Such a shift changes neither the value of the transport cost nor the matrix $F \vc C_E - \sigma^A \otimes \bone_N - \bone_N \otimes \sigma^B$ defined in \eqref{SN}. Consequently, by a unique solution to the dual problem we shall always mean the uniqueness of the matrix $F$.

We now show the uniqueness of the solution to the optimal quantum transport for nonisospectral mixed qubits \footnote{We are indebted to Shmuel Friedland for providing us arguments leading to a shorter proof of this result.}. To this end, we first need the following result.

\begin{lem}\label{lem_rank}
Let  $\rho^A, \rho^B \in \Omega_2$ be of rank two. For such states, let $\ur^{AB} \in \Gamma(\rho^A,\rho^B)$ be an optimal coupling yielding the minimum in \eqref{TQ} and let $\uF$ be the optimal matrix solving the dual problem \eqref{dual}. Then, $\rank \ur^{AB} = \rank \underline{F} = 2$.
\end{lem}
\begin{proof}
This fact is shown in the course of the proof of Theorem 5.1 in \cite{transport2}.
\end{proof}

\begin{thm}\label{thm_uniq_sol}
Let  $\rho^A, \rho^B \in \Omega_2$ be two non-isospectral mixed density matrices. Then, there exists a unique solution to the original and the dual quantum transport problem.
\end{thm}
\begin{proof}
Assume that both $\rho^A, \rho^B$ are of full rank. Lemma \ref{lem_rank} shows that any optimal coupling $\ur^{AB}$ has rank two. Suppose that $\ur^{AB}_{\,1}$ and $\ur^{AB}_{\,2}$ are two different optimal solutions to the original problem. Suppose first that $\ur^{AB}_{\,1}$ and $\ur^{AB}_{\,2}$ have different supports. Then, for any $a \in (0,1)$, the state $\ur^{AB}_{\,a} \vc a \ur^{AB}_{\,1} + (1-a) \ur^{AB}_{\,2}$ is an optimal coupling for $\rho^A, \rho^B$. But $\rank \ur^{AB}_{\,a} \geq 3$, what contradicts Lemma \ref{lem_rank}. Suppose then that $\ur^{AB}_{\,1}$ and $\ur^{AB}_{\,2}$ are both supported by the same two-dimensional subspace. Now, set $\ur^{AB}_{\,t} \vc -t \ur^{AB}_{\,1} + (1+t) \ur^{AB}_{\,2}$ for any $t>0$. Any such $\ur^{AB}_{\,t}$ is an optimal coupling for $\rho^A, \rho^B$.
But since the set of density matrices is bounded, there exists a maximal $t$ such that for each $t' > t$ the coupling $\ur^{AB}_{\,t'}$ has one negative eigenvalue. Hence $\ur^{AB}_{\,t}$ has only one nonzero eigenvalue, which contradicts Lemma \ref{lem_rank}.

The same argument shows the uniqueness of the optimal matrix $\uF$ solving the dual problem.
\end{proof}

We shall now present an explicit formula for the optimal coupling of two commuting qubits. Using the Bloch vectors,
\begin{align}\label{Bloch}
\rho^A = \frac{1}{2} \big( \bone + \vec{a} \cdot \vec{\sigma} \big), && \rho^B \vc \frac{1}{2} \big( \bone + \vec{b} \cdot \vec{\sigma} \big),
\end{align}
it is convenient to write an element of $\Gamma(\rho^A,\rho^B)$ in the Fano form \cite{Fano}
\begin{equation}
\label{eqStates_01}
\rho^{AB} = \frac{1}{4}\left(\bone_4 + \sum_{i = 1}^3 a_i \sigma_i \otimes \bone_2 + \sum_{i = 1}^3 b_i \bone_2 \otimes \sigma_i + \sum_{i,j = 1}^3 R_{ij} \sigma_i \otimes \sigma_j  \right).
\end{equation}
It is straightforward to show that the quantum transport cost depends only on the trace of the correlation matrix,
\begin{align}\label{costR}
\Tr C^Q \rho^{AB} = \tfrac{1}{4} \left( 1 - \Tr R \right).
\end{align}
Note however that the requirement of positive semi-definiteness of $\rho^{AB}$ induces constraints on the matrix $R$ --- see Appendix \ref{app:Fano}. Consequently, the minimisation of \eqref{costR} is performed over a complicated convex subset of the 9-dimensional parameter space of $R$.

If $\rho^A$ and $\rho^B$ commute then their Bloch vectors can be aligned via a unitary transformation and hence we can assume that only $a_3$ and $b_3$ are non-zero. We claim that the correlation matrix $R$ of the optimal coupling $\ur^{AB}$ is diagonal with entries
\begin{equation}
\label{eqStates_14}
\begin{aligned}
& R_{11} = R_{22} = \frac{1}{2} \sqrt{(\vert a_3 + b_3\vert-2)^2 - (a_3 - b_3)^2}, % \\
&  R_{33} = \vert a_3 + b_3\vert -1.
\end{aligned}
\end{equation}
Indeed, by inserting \eqref{eqStates_14} to formula \eqref{costR}, for $a_3 + b_3 > 0$ one gets 
\begin{align*}
\Tr C^Q \underline{\rho}^{AB}
%= \frac{1}{4}\left(1 - 2\sqrt{(1 - a_3)(1 - b_3)} - a_3 - b_3 + 1\right) 
= \frac{1}{4}\left(\sqrt{1 - a_3} - \sqrt{1 - b_3} \right)^2 
\end{align*}
and in the case of $a_3 + b_3 < 0$ one gets 
\begin{align*}
\Tr C^Q \underline{\rho}^{AB} 
%= \frac{1}{4}\left(1 - 2\sqrt{(1 + a_3)(1 + b_3)} + a_3 + b_3 + 1\right) 
= \frac{1}{4}\left(\sqrt{1 + a_3} - \sqrt{1 + b_3} \right)^2.
\end{align*}
Equivalently, we can write
\begin{align*}
\Tr C^Q \underline{\rho}^{AB} 
%= \frac{1}{4}\left(1 - 2\sqrt{(1 + a_3)(1 + b_3)} + a_3 + b_3 + 1\right) 
= \frac{1}{4} \max \left( \left(\sqrt{1 - a_3} - \sqrt{1 - b_3} \right)^2,  \left(\sqrt{1 + a_3} - \sqrt{1 + b_3} \right)^2 \right),
\end{align*}
which agrees with Eq. \eqref{Tcomm} derived in \cite{transport1}.
The optimal coupling for commuting qubit states is of rank 2, as follows from Lemma \ref{lem_rank}, and its eigenvalues are $\tfrac{1}{2} \vert a_3 + b_3\vert$ and $1-\tfrac{1}{2} \vert a_3 + b_3\vert$.

Given the optimal coupling, $\ur^{AB}$ one can compute the optimal observables, $\us^A$ and $\us^B$, saturating the dual problem \eqref{dual}. To do so one may use Theorem 3.2 from \cite{transport2} which states that for any optimal coupling $\underline{\rho}^{AB}$ and any optimal $\underline{F}$ the following equality holds:

\begin{equation}
\label{diff_support}
    \Tr[\underline{F} \underline{\rho}^{AB}] = 0.
\end{equation}

Thanks to the non-negativity of $\underline{\rho}^{AB}$ and $\underline{F}$ this is equivalent to the statement that $\underline{\rho}^{AB}$ and $\underline{F}$ are supported on the orthogonal subspaces of the Hilbert space $\mathbb{C}^4$. 
By inserting the solution of the original problem \eqref{eqStates_14} and a general form of $F$ \eqref{SN} into \eqref{diff_support} we obtained a system of linear equations with the unique solution:

\begin{equation}
\label{dual_com_sol2}
\begin{aligned}
\text{if ~} a_3 + b_3 > 0 \, , & \\
& \us^A = \frac{1}{8} \begin{pmatrix}
\frac{\vert a_3 - b_3\vert}{\sqrt{1-a_3}\sqrt{1-b_3}} & 0\\
0 & 4 - \frac{ 2(2 - a_3 - b_3)+ \vert a_3 - b_3\vert }{\sqrt{1-a_3}\sqrt{1-b_3}}
\end{pmatrix}, \\
& \us^B = \frac{1}{8} \begin{pmatrix}
\frac{-\vert a_3 - b_3\vert}{\sqrt{1-a_3}\sqrt{1-b_3}} & 0\\
0 & 4 - \frac{ 2(2 - a_3 - b_3) - \vert a_3 - b_3\vert }{\sqrt{1-a_3}\sqrt{1-b_3}}
\end{pmatrix}; \\[5 pt]
\end{aligned}
\end{equation}
\begin{equation}
\begin{aligned}
\text{if ~} a_3 + b_3 < 0 \, , & \\
& \us^A = \frac{1}{8} \begin{pmatrix}
4 - \frac{ 2(2 + a_3 + b_3) - \vert a_3 - b_3\vert }{\sqrt{1+a_3}\sqrt{1+b_3}} & 0\\
0 & \frac{ -\vert a_3 - b_3\vert}{\sqrt{1+a_3}\sqrt{1+b_3}}
\end{pmatrix}, \\
& \us^B = \frac{1}{8} \begin{pmatrix}
4 - \frac{ 2(2 + a_3 + b_3) + \vert a_3 - b_3\vert }{\sqrt{1+a_3}\sqrt{1+b_3}} & 0\\
0 & \frac{ \vert a_3 - b_3\vert}{\sqrt{1+a_3}\sqrt{1+b_3}}
\end{pmatrix}. \\
\end{aligned}
\end{equation}

\bigskip

Now, let us discuss the case of isospectral states.
We start by finding an exemplary solution to the quantum transport problem between the isospectral states $\rho^A$, $\rho^B = U \rho^A U^\dagger$. By unitary invariance of transport cost we may assume that state $\rho^A$ is diagonal, and state $\rho^B$ is real, so $U = \exp(i \phi \sigma_y /2)$.

We use the following \textit{anzatz} for the coupling matrix:

\begin{equation}\label{iso_ansatz}
    \rho^{AB} = (\bone_2 \otimes U) \rho^{AA}   (\bone_2 \otimes U^\dagger),
\end{equation}
where $\rho^{AA} = \rho^A \otimes \bone_2 + \bone_2 \otimes \rho^A + \sum_{i=1}^3 R'_{ii} \sigma_i \otimes \sigma_i$ is some coupling between $\rho^A$  and itself with a diagonal correlation matrix $R'$. % and $U$ is a unitary rotation transform state $\rho^A$ into $\rho^B$. 
Therefore the matrix of correlation coefficients $R$ for $\rho^{AB}$ has the form $R =  R' O^T$, where $O$ is the rotation matrix changing the Bloch vector $\vec{a}$ into the Bloch vector $\vec{b}$. Moreover, since $\rho^{AB}$ and $\rho^{AA}$ are connected by unitary map $\rho^{AB} \geq 0 \iff \rho^{AA} \geq 0$.

The transport cost is obtained by minimizing the formula:

\begin{equation}
\label{eqStates_31}
\Tr C^Q \rho^{AB} = \frac{1}{8}(1 + \sqrt{2} R'_- + R'_{33}) (1 - \cos(\phi)),
\end{equation}
with $R_-' = (R_{11}' - R_{22}')/\sqrt{2}$. The minimization of  \eqref{eqStates_31} over $R'_{33}$ and $R'_{-}$ under the constraints imposed by non-negativity of $\rho^{AB}$  (see Eq. \eqref{eqStates_21} in the Appendix) leads to:

\begin{equation*}
T_{C^Q}^Q = \frac{1}{2}\left(1 - \sqrt{1 - a_3^2}\right)\sin(\phi/2)^2,
\end{equation*}
which agrees with Eq. \eqref{Tiso} obtained in \cite{transport1}. The correlation matrix $R$ determining the optimal coupling $\underline{\rho}^{AB}$ has the following form:

\begin{equation}
\label{eqStates_33}
R = \begin{pmatrix}
-\sqrt{1 -  a_3^2} \cos(\phi) & 0 & \sqrt{1 -  a_3^2} \sin(\phi)\\
0 & \sqrt{1 - a_3^2} & 0 \\
\sin(\phi) & 0 & \cos(\phi)
\end{pmatrix}
\end{equation}

The optimal coupling $\rho^{AB}$ has rank one and thus can be written as $\rho^{AB} = \vert \psi \rangle\langle \psi \vert$ with

\begin{equation*}
\vert \psi \rangle = \frac{1}{\sqrt{2}}
\begin{pmatrix}
\sqrt{1+a_3} \cos(\phi/2)\\
\sqrt{1+a_3} \sin(\phi/2)\\
\sqrt{1-a_3} \sin(\phi/2)\\
-\sqrt{1-a_3}\cos(\phi/2)
\end{pmatrix}~.
\end{equation*}

Since the construction of solutions for the dual problem, based on orthogonality of $\underline{F}$ and $\underline{\rho}^{AB}$ supports, is straightforward but tedious, we present it in the Appendix \ref{app:izo_dual}. We found a one-parameter family of solutions for the dual problem, such that for each $\phi$ there exist at least one optimal $\underline{F}$ with rank 3.
Eq. \eqref{diff_support} and the fact that both $\underline{F}$ and $\underline{\rho}^{AB}$ are positive semidefinite, together imply that for each pair of isospectral states $\rho^A$ and $\rho^B$, there exist only one optimal $\underline{\rho}^{AB} = \vert \psi\rangle\langle \psi\vert$, with $\vert\psi\rangle$ being the unique eigenvector to $0$ eigenvalue of above mentioned optimal $\underline{F}$. Hence, the original quantum optimal transport problem possesses only one solution, determined by the correlation matrix \eqref{eqStates_33}.

%%%%%%%%%%%%
%%%%%%%%%%%%
%%%%%%%%%%%%
%%%%%%%%%%%%
%%%%%%%%%%%%
%%%%%%%%%%%%
%%%%%%%%%%%%
%%%%%%%%%%%%
%%%%%%%%%%%%
%%%%%%%%%%%%
%%%%%%%%%%%%
%%%%%%%%%%%%
%%%%%%%%%%%%
%%%%%%%%%%%%
%%%%%%%%%%%%
%%%%%%%%%%%%

\textit{Acknowledgements:} It is a pleasure to thank Shmuel Friedland and Tomasz Miller for their helpful remarks. Financial support by Narodowe Centrum Nauki under the Maestro grant number DEC-
2015/18/A/ST2/00274 and by Foundation for Polish Science under the Team-Net
project no. POIR.04.04.00-00-17C1/18-00 is gratefully acknowledged.

\begin{appendices}

\section{Proofs of technical results}\label{sec:proofs}

\subsection{Semi-analytic formulae for $T^Q$ for $N=2$}

In this Appendix we present the proof of Proposition \ref{prop:F}, along with the semi-analytic and analytic formulae for the optimal quantum transport cost for qubits, which we invoked in Section \ref{sec:Riemman}.

We first recall a semi-analytic formula for $T^Q$ between two single-qubit states in the Bloch parametrisation \eqref{rho} derived in \cite{transport2}:
\begin{multline}\label{tq_org}
T^Q \big(\rho(s,0),\rho(r,\theta) \big)\\
 = \max_{\phi \in [0, 2\pi)} \tfrac{1}{4} \left(\sqrt{1+ (2 s-1) \cos \phi }-\sqrt{1+ (2 r  -1) \cos (\theta + \phi)}\right)^2.
\end{multline}
For the special case of commuting or isospectral states one can derive \cite{transport1} explicit analytic formulae:
\begin{align}\label{Tcomm}
T^Q \big(\rho(s,0),\rho(r,0) \big) & = \tfrac{1}{2} \max \left\{\left(\sqrt{r} - \sqrt{s}\right)^2, \left(\sqrt{1-r} - \sqrt{1-s}\right)^2 \right\},\\  \label{Tiso}
T^Q \big(\rho(r,0),\rho(r, \theta) \big) & = \left(\tfrac{1}{2} - \sqrt{r(1-r)} \right) \sin^2(\theta/2).
\end{align}

We shall now prove Proposition \eqref{prop:F} via the following lemma:

\begin{lem}
For any $r \in (0,1)$ and any tangent vector $v = (v_1,v_2) \in T_{\rho(r,0)}\Omega_2^\mathbb{R}$ we have
\begin{align}\label{TQt}
T^Q \big(\rho(r,0),\rho(r + v_1 t,v_2 t) \big) = G(r,v) \, t^2 + O(t^3), \quad \text{ as } t \downarrow 0,
\end{align}
where the function $G$ is defined by formula \eqref{our_F},
\begin{align*}
G(r,v) = \max_{\phi \in [0,2\pi)}  \frac{\big( 2 v_1 \cos(\phi) - (2r - 1) v_{2}\sin(\phi)  \big)^2}{16 \big( 1 + (2r-1) \cos(\phi) \big)}.
\end{align*}
\end{lem}

\begin{proof}
In order to slightly simplify the notation we set $\xi \vc 2 r - 1 \in (-1,1)$ and assume that $\xi \neq 0$, i.e. $r \neq 1/2$. Note that the maximally mixed state $\rho(1/2,0)$ commutes with any other state $\rho^B$ in the Bloch ball, hence one can use Formula \eqref{Tcomm} to compute the value of $T^Q\big(\rho(1/2,0),\rho^B \big) $ (cf. Eq. (12) in \cite{transport1}).

Let us start with writing, for sufficiently small $t > 0$,
\begin{align}
W \big(\rho(r,0),\rho(r + v_1 t,v_2 t) \big) & = \tfrac{1}{2} \max_{\phi \in [0,2\pi)} \left\vert \sqrt{1 + \xi \cos \phi} - \sqrt{1 + (\xi + 2 v_1 t) \cos (\phi + v_2 t) } \right\vert \notag \\
& \leq \sqrt{G(r,v)} \, t + \tfrac{1}{2} \, \max_{\phi \in [0,2\pi)} R(\xi,v,\phi,t), \label{WR}
\end{align}
where
\begin{align}
& R(\xi,v,\phi,t) \label{R} \\
& \qquad \vc \left\vert \sqrt{1 + (\xi + 2 v_1 t) \cos (\phi + v_2 t) } - \sqrt{1 + \xi \cos \phi} - \frac{2 v_1 t \cos \phi - \xi v_2 t \sin \phi}{2 \sqrt{1 + \xi \cos \phi}} \right\vert. \notag
\end{align}

Our goal is to establish a uniform bound on $R(\xi,v,\phi,t)$, which is at least quadratic in $t$ and does not depend on $\phi$. 

Let us identify the three terms in Formula \eqref{R},
\begin{align*}
A & \vc 1 + (\xi + 2 v_1 t) \cos (\phi + v_2 t),  && B \vc  1 + \xi \cos \phi, \\
C & \vc 2 v_1 t \cos \phi - \xi v_2 t \sin \phi.
\end{align*}
Then, we have
\begin{align*}
R & = \left\vert \sqrt{A} - \sqrt{B} - \frac{C}{2 \sqrt{B}} \right\vert = \left\vert \frac{A - B - C}{\sqrt{A} + \sqrt{B}} + \frac{C(B - A)}{2 \sqrt{B} (\sqrt{A} + \sqrt{B})^2} \right\vert \\
& \leq \frac{\vert A - B - C\vert}{\sqrt{B}} + \frac{\vert C \vert \cdot \vert A - B\vert}{2 B^{3/2}}.
\end{align*}
The most obvious estimates are
\begin{align*}
B \geq 1 - \vert \xi \vert, && \vert C \vert \leq \big( 2 \vert v_1 \vert + \vert \xi v_2 \vert \big) t.
\end{align*}
In order to establish suitable bounds on the remaining terms we will use the sum-to-product trigonometric identities \cite{AS72}. We have
\begin{align*}
\vert A - B - C \vert & = \Big\vert (\xi + 2 v_1 t) \big[  \cos (\phi + v_2 t) -  \cos \phi \big ] + \xi v_2 t \sin \phi \Big\vert \\
& =  \Big\vert 2 (\xi + 2 v_1 t) \sin \left( \phi + \tfrac{v_2 t}{2} \right) \sin \left( \tfrac{v_2 t}{2} \right) - \xi v_2 t \sin \phi \Big\vert \\
& \leq 2 \vert \xi \vert \cdot \Big\vert \sin \left( \phi + \tfrac{v_2 t}{2} \right) \sin \left( \tfrac{v_2 t}{2} \right) - \tfrac{1}{2} \xi v_2 t \sin \phi \Big\vert + 4 t \vert v_1 \vert \cdot \big\vert \sin \left( \tfrac{v_2 t}{2} \right) \big\vert  \\
& \leq 2 \vert \xi \vert \cdot \Big[ \big\vert \sin \left( \phi + \tfrac{v_2 t}{2} \right) - \sin \phi \big\vert \cdot \big\vert \sin \left( \tfrac{v_2 t}{2} \right) \big\vert + \\
& \qquad\qquad\qquad\qquad\qquad + \vert \sin \phi \vert \cdot \big\vert \sin \left( \tfrac{v_2 t}{2} \right) - \tfrac{v_2 t}{2} \big\vert \Big] + 2 \vert v_1 v_2 \vert t^2 \\
& \leq 2 \vert \xi \vert \cdot \Big[ \big\vert \sin \left( \tfrac{v_2 t}{4} \right) \big\vert \cdot \big\vert \cos \left( \phi + \tfrac{v_2 t}{4} \right) \big\vert \cdot \vert v_2 \vert t + \tfrac{\vert v_2 \vert^3 t^3}{48} \Big] + 2 \vert v_1 v_2 \vert t^2 \\
& \leq \tfrac{1}{2} \big( \vert \xi v_2^2 \vert + 4 \vert v_1 v_2 \vert \big) t^2 + \tfrac{1}{24} \vert v_2 \vert^3 t^3. 
\end{align*}
Similarly, we obtain
\begin{align*}
\vert A - B \vert & = \Big\vert \xi \big[  \cos (\phi + v_2 t) -  \cos \phi \big ] + 2 v_1 t \, \cos (\phi + v_2 t) \Big\vert \\
& \leq 2 \vert \xi \vert \cdot \big\vert \sin \left( \phi + \tfrac{v_2 t}{2} \right) \big\vert \cdot \big\vert \sin \left( \tfrac{v_2 t}{2} \right) \big\vert + 2 t \vert v_1 \vert \\
& \leq  \big( \vert \xi v_2 \vert + 2 \vert v_1 \vert \big) t.
\end{align*}

Summa summarum, we arrive at the following estimate
\begin{align*}
R(\xi,v,\phi,t) \leq \left( \frac{ \vert \xi v_2^2 \vert + 4 \vert v_1 v_2 \vert }{2\sqrt{1 - \vert \xi \vert}} + \frac{ \big( \vert \xi v_2 \vert + 2 \vert v_1 \vert \big)^2 }{(1 - \vert \xi \vert)^{3/2}} \right) t^2 + \frac{ \vert v_2 \vert^3 }{24\sqrt{1 - \vert \xi \vert}} \, t^3.
\end{align*}
This provides the desired uniform bound, which we can insert into Formula \eqref{WR} and conclude that
\begin{align*}
W \big(\rho(r,0),\rho(r + v_1 t,v_2 t) \big) & \leq \sqrt{G(r,v)} \, t + c_1 t^2 + c_2 t^3.
\end{align*} 
The square of the above equality implies Eq. \eqref{TQt}.
\end{proof}

One can show (cf. \cite[Appendix B]{transport2}) that the angle $\phi$ yielding the maximum in Eq. \eqref{our_F} defining the function $G$ can be obtained from the solution to a quartic equation. For the specific cases of two commuting or isospectral qubits we have
\begin{align}
\label{g11_g22_values}
G\big((r,(1,0)\big) = \tfrac{1}{8} \max \big\{ \tfrac{1}{1 - r}, \tfrac{1}{r} \big\}, && G\big((r,(0,1)\big) = \tfrac{1}{8} \left(1 - 2 \sqrt{r(1-r)}\right).
\end{align}
Equipped with these formulae one can reproduce Fig. \ref{fig:Riem} and convince oneself that Eq. \eqref{gF} fails.

\subsection{Proof of Lemma \ref{lem_F}}

In this Appendix we present the proof of Lemma \ref{lem_F}, which is the key to Theorem \ref{thm:qBit_monoton}. Let us fix a single-qubit channel $\Phi$ determined by two Kraus operators \eqref{k2}. The channel, as well as its dual, $\Phi^*(\sigma) = \sum_{i=1,2} K_i^\dagger \sigma K_i$, is parametrised by two angles $u \in [0,2\pi)$ and $v \in [0,\pi)$.   

Lemma \ref{lem_F} claims that if $\sigma^A$ and $\sigma^B$ are two diagonal Hermitian 2 $\times$ 2 matrices, then the positivity of the operator $F = C^Q - \sigma^A \otimes \bone - \bone \otimes \sigma^B$, appearing in the dual optimisation problem \eqref{dual}, implies the positivity of the operator $F^\Phi = C^Q - \Phi^*(\sigma^A) \otimes \bone - \bone \otimes \Phi^*(\sigma^B)$, for the extremal single-qubit channel $\Phi$ determined by Eq. \eqref{k2}.

Since both matrices $\sigma^A$ and $\sigma^B$ are diagonal, we have four independent real parameters. Note, however, that  a gauge transformation $\sigma^A \mapsto \sigma^A + a~ \bone$, $\sigma^B \mapsto \sigma^B - a~ \bone$, with any $a \in \fR$, does not affect neither the operator $F$, nor the value of the optimal quantum transport cost \eqref{dual}. Consequently, we can eliminate one of these four parameters. It is convenient to use the following parametrisation: 
\begin{align*}
  c \vc \Tr \big( \sigma^A \sigma_3 \big) \,, &&   d \vc \Tr \big( \sigma^B \sigma_3 \big) \, , && x \vc \Tr \big( \sigma^A + \sigma^B \big) ,
\end{align*}
where $\sigma_3$ denotes the third Pauli matrix. Then, both hermitian matrices $F$ and $F^{\Phi}$ take the following form

\begin{align}
F =  \frac{1}{2} \, \begin{pmatrix}
F_{11} & 0 & 0 & 0 \\
0 & F_{22} & F_{23} & 0 \\
0 & F_{23} & F_{33} & 0 \\
0 & 0 & 0 & F_{44}	
\end{pmatrix} , &&
F^{\Phi} =   \frac{1}{2} \begin{pmatrix}
F^{\Phi}_{11} & 0 & 0 & 0 \\
0 & F^{\Phi}_{22} & F^{\Phi}_{23} & 0 \\
0 & F^{\Phi}_{23} & F^{\Phi}_{33} & 0 \\
0 & 0 & 0 & F^{\Phi}_{44}	
\end{pmatrix} ,
\end{align}

with entries:

\begin{equation}\label{Fmatrix}
\begin{aligned}
F_{11} & = - \left(c + d + x \right), &  F_{22} & = 1-c+d- x, & F_{23} = F_{32} = -1, \\
F_{33} & = 1+ c-d- x, & F_{44} & = c+d- x,
\end{aligned}
\end{equation}
\begin{equation}\label{Fphi}
\begin{aligned}
F_{11}^{\Phi} & = - \left(c + d \right) \cos (u-v) - x, & \\
F_{22}^{\Phi} & = 1 - c \cos(u-v) + d \cos(u+v) - x, & F_{23}^{\Phi} = F_{32}^{\Phi} = -1, \\
F_{33}^{\Phi} & = 1 + c \cos(u+v) - d \cos(u-v) - x, & \\
F_{44}^{\Phi} & = \left(c + d \right) \cos (u+v) - x.
\end{aligned}
\end{equation}

The demand $F \geq 0$ implies, in particular, the non-negativity of the diagonal elements of $F$, which yields the following constraints on the parameters,
\begin{align}\label{parF}
x \leq 0, && x \leq c + d \leq -x, && x - 1 \leq c - d \leq 1 - x.
\end{align}
The non-negativity of the central 2 $\times$ 2 minor of $F$, which is the only non-trivial one, gives an additional constraint
\begin{align}\label{parF2}
(c - d)^2 \leq x(x-2).
\end{align}

To simplify notation let us set
\begin{align*}
\alpha \vc \cos (u-v) \quad \text{ and } \quad \beta \vc \cos (u+v).
\end{align*}

We start with showing that the diagonal elements of $F^\Phi$ are non-negative. If $\alpha \in [0,1]$ then \eqref{parF} yields
\begin{equation*}
F^{\Phi}_{11} = - \left(c + d \right) \alpha - x \geq x \big( \alpha - 1 \big) \geq 0,
\end{equation*}
while if $\alpha \in [-1,0]$ then we obtain
\begin{equation*}
F^{\Phi}_{11} \geq - x ( \alpha + 1 ) \geq 0.
\end{equation*}
Analogously, we deduce that $F^{\Phi}_{44} \geq 0$.

Let us now rewrite
\begin{align*}
F_{22}^{\Phi} & = 1 - x - \tfrac{1}{2} (c + d) ( \alpha - \beta ) + \tfrac{1}{2} (d-c) ( \alpha + \beta )
\end{align*}
and assume that $(\alpha - \beta) \geq 0$ and $(\alpha + \beta) \geq 0$. Then, inequalities \eqref{parF} yield,
\begin{align*}
F_{22}^{\Phi} & \geq 1 - x + \tfrac{1}{2} x ( \alpha - \beta ) + \tfrac{1}{2} (x-1) ( \alpha + \beta ) = \big[ 1 - \tfrac{1}{2} ( \alpha + \beta ) \big] + \tfrac{1}{2} x ( \alpha - 1) \geq 0,
\end{align*}
because $\alpha \in [0,1]$ and $(\alpha + \beta) \in [0,2]$. Analogously, one shows that $F_{22}^{\Phi} \geq 0$ under three other possible assumptions about the signs of $\alpha - \beta$ and $\alpha + \beta$. Along the same lines, one can prove that $F_{33}^{\Phi} \geq 0$.

Finally, let us consider the central 2 $\times$ 2 minor of $F^\Phi$. With $\pm$ denoting the sign of $\alpha + \beta$ and using the previous estimates on $F_{22}^{\Phi}$ and $F_{33}^{\Phi}$, along with constraint \eqref{parF2}, we can write the determinant of this minor as
\begin{align*}
F_{22}^{\Phi} F_{33}^{\Phi} - 1 & = \Big[ 1 - x - \tfrac{1}{2} (c + d) ( \alpha - \beta ) - \tfrac{1}{2} (c-d) ( \alpha + \beta ) \Big] \times \\
& \qquad \times \Big[ 1 - x - \tfrac{1}{2} (c + d) ( \alpha - \beta ) + \tfrac{1}{2} (c-d) ( \alpha + \beta ) \Big] - 1 \\
& \geq \Big[ 1 - x \pm \tfrac{1}{2} x ( \alpha - \beta ) - \tfrac{1}{2} (c-d) ( \alpha + \beta ) \Big] \times \\
& \qquad \times \Big[ 1 - x \pm \tfrac{1}{2} x ( \alpha - \beta ) + \tfrac{1}{2} (c-d) ( \alpha + \beta ) \Big] - 1 \\
& = \Big[1 - x \Big(1 \mp \tfrac{1}{2} ( \alpha - \beta ) \Big) \Big]^2 - \tfrac{1}{4} (c-d)^2 (\alpha + \beta)^2 - 1 \\
& \geq - 2x \Big(1 \mp \tfrac{1}{2} ( \alpha - \beta ) \Big) + x^2 \Big(1 \mp \tfrac{1}{2} ( \alpha - \beta ) \Big)^2  - \tfrac{1}{4} x(x-2) (\alpha + \beta)^2 \\
%& = x^2 \Big[ 1 \mp (\alpha - \beta) + \tfrac{1}{4} (\alpha - \beta)^2 - \tfrac{1}{4}(\alpha + \beta)^2 \Big] - x \Big[ 2 \mp (\alpha - \beta) + \tfrac{1}{2} (\alpha + \beta)^2 \Big] \\
& = x^2 (1 \mp \alpha) (1 \pm \beta)  - x \Big[ \tfrac{1}{2} (1 \mp \alpha)^2 + \tfrac{1}{2} (1 \pm \beta)^2 + 1 + \alpha \beta \Big] \geq 0.
\end{align*}
The last inequality holds because $\alpha, \beta \in [-1,1]$, while $x \leq 0$ from constraints \eqref{parF}.

Since the central minor of $F^\Phi$ is the only non-trivial one we conclude that, indeed,  if $F$ is positive semidefinite, then $F^\Phi$ is so, for any quantum channel determined by Kraus operators \eqref{k2}.

%%%%%%%%%%%
%%%%%%%%%%%
%%%%%%%%%%%
%%%%%%%%%%%
%%%%%%%%%%%
%%%%%%%%%%%
%%%%%%%%%%%
%%%%%%%%%%%

\subsection{Conditions for positive semi-definiteness in Fano form}
\label{app:Fano}

In this Appendix, we derive a useful form of the constraint for the coupling matrix between two qubits to be positive semidefinite, using its Fano form \cite{Fano}. The obtained results are based on general formulas from \cite{byrd}.

The density matrix representing any quantum state must be Hermitian, positive semidefinite and of trace $1$. 
The first and the last property can be expressed easily using parametrisation via generators of the $su(n)$ Lie algebras, i.e., Hermitian traceless matrices. This leads to the Fano representation of a bipartite quantum state $\rho^{AB}$, which for two qubits takes the following form:

\begin{equation*}
    \rho^{AB} = \frac{1}{4}\left(\bone_4 + \sqrt{6} \vec{n} \cdot \vec{\lambda} \right)~,
\end{equation*}
where we temporarily adapt the normalization form \cite{byrd}, and use the basis of Hermitian operators on $\fC^{2\times 2}$ defined by the tensor product of Pauli operators:

\begin{equation*}
\begin{aligned}
& \lambda_i ,~ i = 1,2,3 &\leftrightarrow~~ &\frac{1}{\sqrt{2}} \sigma_j \otimes \bone_2,\\
& \lambda_i ,~~ i = 4,5,6 &\leftrightarrow~~ &\frac{1}{\sqrt{2}} \bone_2 \otimes \sigma_j, \\
& \lambda_i ,~~ i = 7,8,9 &\leftrightarrow~~ &\frac{1}{\sqrt{2}} \sigma_1 \otimes \sigma_j, \\
& \lambda_i ,~~ i = 10,11,12 &\leftrightarrow~~ &\frac{1}{\sqrt{2}} \sigma_2 \otimes \sigma_j, \\
& \lambda_i ,~~ i = 13,14,15 &\leftrightarrow~~ &\frac{1}{\sqrt{2}} \sigma_3 \otimes \sigma_j, \\
\end{aligned}
\end{equation*}

with $j = 1,2,3$ and the vector $\vec{n}$ has the form:

\begin{equation*}
    \vec{n} = ( a_1, a_2, a_3, b_1, b_2,b_3, R_{11}, R_{12}, R_{13}, R_{21}, R_{22}, R_{23}, R_{31}, R_{32}, R_{33})~.
\end{equation*}
Here $\vec{a}$, $\vec{b}$ are proportional to the Bloch vectors for consecutive subsystems and the matrix $R$ is defined in Eq. \eqref{eqStates_01}.
The density matrix $\rho^{AB}$ has non-negative eigenvalues if and only if all coefficients in its characteristic polynomial are non-negative \cite{byrd},

\begin{equation*}
    \det[\rho^{AB} - \lambda \mathbb{I}] = \lambda^4 - S_1 \lambda^3 + S_2 \lambda^2 - S_3 \lambda + S_4 = 0~.
\end{equation*}

In \cite{byrd} it was shown tha these coefficients can  be written as:
\begin{equation}
\label{S_intro}
\begin{aligned}
& S_1 = 1~, \\
& S_2 = \frac{3}{8}(1 - \vec{n} \cdot \vec{n})~, \\
& S_3 = \frac{1}{16}\left( 1 - 3 \vec{n} \cdot \vec{n} + 2 \vec{n} \cdot (\vec{n} \star \vec{n})  \right)~, \\
& S_4 = \frac{1}{64}\left(1 - 6 \vec{n} \cdot \vec{n} + 8 \vec{n} \cdot (\vec{n} \star \vec{n}) + 9 (\vec{n} \cdot \vec{n})^2  - 12  (\vec{n} \star \vec{n}) \cdot (\vec{n} \star \vec{n})  \right)~,
\end{aligned}
\end{equation}
where $\star$ denotes the product defined  \cite{byrd} via symmetric structure constants $d_{ijk}$ of the Lie algebra $su(4)$,
\begin{equation*}
    (\vec{a} \star \vec{b})_k = \frac{\sqrt{6}}{2} d_{ijk} a_i b_j~.
\end{equation*}
The expression $\vec{n} \cdot (\vec{n} \star \vec{n}) $ has a compact form
\begin{equation*}
    \vec{n} \cdot (\vec{n} \star \vec{n})  = 3  \sqrt{3} \left( \vec{a}^T R \vec{b} - \text{det}(R)  \right)~,
\end{equation*}
while the expression  $(\vec{n} \star \vec{n}) \cdot (\vec{n} \star \vec{n}) $ can be simplified to
\begin{multline*}
    (\vec{n} \star \vec{n}) \cdot (\vec{n} \star \vec{n}) = 3 \bigg[ \vec{b}^T R^T R \vec{b} + \vec{a}^T R R^T\vec{a}   + \\
    + (\text{det} R)^2 \sum_{i,j} \big((R^{-1})_{ij}\big)^2 - 2 (\text{det} R)~ \vec{b}^T R^{-1}\vec{a} + \vec{a}^2\vec{b}^2  \bigg]~,
\end{multline*}
under the assumption that the matrix $R$ is invertible.
Hence the coefficients $S_1, S_2, S_3, S_4$ entering \eqref{S_intro} can be rewritten as:

\begin{equation}
\label{eqBurd_01}
\begin{aligned}
& S_1 = 1~, \\
& S_2 = \frac{3}{8}\bigg(1 - \sum_{i,j} (R_{ij})^2 - \vec{a}^2 - \vec{b}^2 \bigg)~, \\
& S_3 = \frac{1}{16}\left[ 1 - 3 \bigg(\sum_{i,j} (R_{ij})^2 + \vec{a}^2 + \vec{b}^2\bigg) + 6 \sqrt{3}\bigg( \vec{a}^T R \vec{b} - (\text{det} R)  \bigg) \right]~, \\
& S_4 = \frac{1}{64}\Bigg[1 - 6 \bigg(\sum_{i,j} (R_{ij})^2 + \vec{a}^2 + \vec{b}^2\bigg) + 24 \sqrt{3}\bigg( \vec{a}^T R \vec{b} - (\text{det} R)  \bigg)  + \\
& \qquad\qquad +9 \bigg(\sum_{i,j} (R_{ij})^2 + \vec{a}^2 + \vec{b}^2\bigg)^2  - 36  \bigg( \vec{b}^T R^T R \vec{b} + \vec{a}^T R R^T\vec{a} + \\
& \qquad\qquad\qquad + (\text{det} R)^2 \sum_{i,j} \big((R^{-1})_{ij}\big)^2 - 2 (\text{det} R)~ \vec{b}^T R^{-1}\vec{a} + \vec{a}^2\vec{b}^2  \bigg) \Bigg]~. \\
\end{aligned}
\end{equation}

The necessary and sufficient conditions for $\rho^{AB}$ to be a density matrix come down to
\begin{equation*}
     0 \leq S_l, \quad \text{ for } \quad l = 1,2,3,4.
\end{equation*}

We end this section by discussing the special case of $\rho^A = \rho^B$, i.e. $\vec{a} = \vec{b} = (0,0,a_3)$, and a diagonal correlation matrix $R = \diag(R_{11}, R_{22}, R_{33})$. Adopting the normalization from \eqref{eqStates_01}, the  expression form \eqref{eqBurd_01} for $S_3$ and $S_4$ simplifies to the product of two planes and a hyperbola:
\begin{equation}
\label{eqStates_21}
\begin{aligned}
& S_3 = (R_{33}-1) \left(2 a_3^2-R_{33}+R_-^2-1\right)-(R_{33}+1) R_+^2 \geq 0~, \\
& S_4 = (1 - R_{33} + \sqrt{2} R_+)(1 - R_{33} -\sqrt{2} R_+) \left(4 a_3^2 + 2 R_-^2 - (1 + R_{33})^2\right) \geq 0~,
\end{aligned}
\end{equation}
where $R_\pm = (R_{11} \pm R_{22})/\sqrt{2}$. %, $R_- = (R_{11}-R_{22})/\sqrt{2}$.
%The conditions for $S_3$ and $S_2$ are of less importance, since they only state that one should choose a central region restricted by \eqref{eqStates_21}.

\subsection{Solutions of dual quantum transport problem for isospectral qubit states}
\label{app:izo_dual}

In this Appendix we present the construction of a one parametric family of solutions to dual quantum optimal transport problem for isospectral qubit states. Afterwards, we use those solutions to show that the solution of the original transport problem \eqref{eqStates_33} is unique. 

For isospectral states $\rho^A$, $\rho^B$ the optimal coupling $\underline{\rho}^{AB}$, given by \eqref{eqStates_33}, has only one nonzero eigenvalue, which demonstrates that Theorem \ref{thm_uniq_sol} cannot hold and we can expect multiple solutions both to the original and the dual problem.

The construction of the solutions to the dual problem for the isospectral states is based on the fact that any optimal coupling $\underline{\rho}^{AB}$ has disjoint support with any optimal $\underline{F}$, as follows from Eq. \eqref{diff_support}. 
Hence, by the eigendecomposition of $\underline{\rho}^{AB}$ we know that the optimal $\underline{F}$ must have the form
\begin{equation}
\label{eq: F_form1}
\underline{F} = \sum_{i,j} \vec{v_i}\vec{v_j}^\dagger s_{i,j}~,
\end{equation}
where $\vec{v_{i}}$ are eigenvectors to zero eigenvalues of $\underline{\rho}^{AB}$ and $s_{ij} = s_{ji}^*$ are some a priori unknown coefficients.
On the other hand, the general form of $F$ is given by 
\begin{equation}
\label{eq: F_form2}
F = \left(\frac{1}{4} - x_0\right)\bone_4 - \sum_{i = 1}^3 c_i \sigma_i \otimes \bone_2 - \sum_{i = 1}^3 d_i \bone_2 \otimes \sigma_i - \frac{1}{4} \sum_{i = 1}^3 \sigma_i \otimes \sigma_i~,
\end{equation}
where $c_i =\frac{1}{2} \Tr[\sigma^A \sigma_i]$, $d_i =\frac{1}{2} \Tr[\sigma^B \sigma_i]$, $x_0 = \Tr[\sigma^A] + \Tr[\sigma^B]$, and $\sigma_i$ are Pauli matrices.

Now, we have to compare  \eqref{eq: F_form1} with \eqref{eq: F_form2} to eliminate as many coefficients as possible and then check when the resulting matrix $\underline{F}$ is positive semidefinite. Elementwise comparison gives the following form of optimal couplings $\underline{\sigma}^A$, $\underline{\sigma}^B$

\begin{equation*}
\begin{aligned}
& (\underline{\sigma}^A)_{11} = {  \frac{a_3 \left(2w+2 a_3-1\right)-2w \left(4 a_3 d_3+1-w\right) \sec \phi -2(a_3-1) \left(w-1\right) \cos \phi}{8 a_3 w}}~, \\
& (\underline{\sigma}^A)_{12} =  \frac{\sqrt{1-a_3} \left(4 a_3 d_3+1-w\right) \tan \phi }{4 a_3\sqrt{1+a_3}}+ \\ & ~~~~~~~~~ - \frac{(\cos \phi +1) \tan (\phi/2)  \left(4 w a_3 d_3-(a_3+1) \left(w-1\right) \cos \phi+w-w^2\right)}{4 a_3(1+a_3) \cos \phi}~,  \\
& (\underline{\sigma}^A)_{21} = (\underline{\sigma}^A)_{12}~, \\
& (\underline{\sigma}^A)_{22} = - (\underline{\sigma}^A)_{11} + \frac{(w-1)(1 - \cos\phi)}{2 w}~, \\
\end{aligned}
\end{equation*}

\begin{equation}
\begin{aligned}
& \underline{\sigma}^B = \frac{1}{4a_3}\left(
\begin{array}{cc}
 4 a_3 d_3 & (4 a_3 d_3-w+1)\tan \phi \\
 (4 a_3 d_3-w+1) \tan \phi & - 4 a_3 d_3 \\
\end{array}
\right)~,
\end{aligned}
\end{equation}
where $w = \sqrt{1-a_3^2}$ and $d_3$ is the only parameter that has not been eliminated. 
For generic values of $d_3$ parameter, those solutions correspond to optimal $\underline{F}$ which has a  rank equal to three.

To ensure the positivity of $\underline{F}$, we check the values of coefficients $S_1, S_2,S_3, S_4$ of the characteristic polynomial of $\underline{F}$, similarly as in Appendix \ref{app:Fano},
\begin{equation*}
    \det[\underline{F} - \lambda \mathbb{I}] = \lambda^4 - S_1 \lambda^3 + S_2 \lambda^2 - S_3 \lambda + S_4~.
\end{equation*}

Firstly, let us note that $S_1 = \Tr[F] = (1 - 4 x_0)  \geq  0$ if and only if $(1-w)\cos (\phi ) \leq 1$, which is true for any  values of $w = \sqrt{1-a_3^2}$ and $\phi$. Since $\underline{F}$ has at most $3$ nonzero eigenvalues, the coefficient $S_4 = \det \underline{F}$ must be equal to 0. The $S_2$ coefficient imposes a quadratic condition in $d_3$:
\begin{multline}
\label{s2dual_izo}
%\begin{aligned}
- 32 d_3^2 a_3^2 w^3  \sec ^2(\phi)  -8 d_3 a_3 w^2  (\sec \phi -1) \left(2 w(1-w) \sec \phi -(1-w)^2\right) + \\ 
 +\sin ^2\left(\phi/2\right) \sec^2(\phi) w(1-w)^2 \Big(2 \left(1+(2-w)w\right) -a_3^2 \cos3 \phi + \\
+ 2  (1+w)^2 \cos2\phi+ (w(4+3w)-3) \cos \phi \Big)\geq 0~,
%\end{aligned}
\end{multline}
and $S_3$ imposes a quadratic condition in $d_3$ as well
\begin{multline}
\label{s3dual_izo}
%\begin{aligned}
  8 d_3^2 a_3^2 w^3 \sin ^2\left(\phi/2\right) \sec ^2(\phi ) + \sin ^4(\phi/2) \sec ^2(\phi ) w(1-w)^2(w- \cos (\phi ) ) + \\
  + 4 d_3 a_3 \sin ^4(\phi /2) \sec (\phi )w(1-w) \big(w\left(\sec (\phi )+1\right) - 1 \big) \leq 0~.
%\end{aligned}
\end{multline}

For both these quadratic inequalities, the discriminants are greater than $0$. Therefore, there exist intervals of $d_3$ satisfying each of them separately. Moreover, the roots of \eqref{s2dual_izo} lie between the roots of \eqref{s3dual_izo}. Hence, introducing the auxiliary variables
\begin{align*}
 d_{3,1} & = \frac{1}{16 a_3 w} \Bigg[(w-1)(1+3w) -(1-w)^2 \cos2\phi + \\
 & \quad + 2 a_3 \cos\phi \bigg(a_3-\frac{\sqrt{2}}{\vert a_3\vert} \left\vert\sin\left(\phi/2\right)\right\vert \sqrt{ (1-w)^3(1+3w) -(1-w)^4 \cos \phi }\bigg) \Bigg]   \\
 d_{3,2} & = \frac{1}{16 a_3 w}  \Bigg[(w-1)(1+3w) -(1-w)^2 \cos 2 \phi + \\
 & \quad + 2 a_3 \cos\phi \bigg(a_3+\frac{\sqrt{2}}{\vert a_3\vert} \left\vert\sin\left(\phi/2\right)\right\vert \sqrt{ (1-w)^3(1+3w) -(1-w)^4 \cos \phi }\bigg) \Bigg] 
\end{align*}
we obtain the following bounds for the value of $d_3$ ,% coefficient 
\begin{equation}
\label{Izo_dual_d3}
\begin{aligned}
&\text{if } \phi \leq \pi/2 &\text{ then} \quad d_{3,1} \leq d_3\leq d_{3,2}~,\\ 
&\text{if } \phi \geq \pi/2 &\text{ then} \quad d_{3,1} \geq d_3\geq d_{3,2}~.\\ 
\end{aligned}
\end{equation}

In the singular case, $\phi = \pi/2$, the entire interval of allowed values of $d_3$ shrinks to a single point: 

\begin{equation*}
    d_3 = d_{3,1} = d_{3,2} = \frac{w-1}{4 a_3}
\end{equation*}
resulting in only a single solution of the dual problem, with rank$(\underline{F}) = 3$ for any value of $a_3 \in (0,1)$. 

We constructed a family of the dual problem solutions in which for any values of $\phi$ and $a_3$ there exists an optimal $\underline{F}$ with rank $3$. Therefore, the corresponding optimal coupling $\underline{\rho}^{AB}$ is uniquely defined by the eigenvector to zero eigenvalue of the abovementioned $\underline{F}$, hence the original problem has only one solution \eqref{eqStates_33}.

%%%%%%%%%
%%%%%%%%%
%%%%%%%%%
%%%%%%%%%
%%%%%%%%%
%%%%%%%%%
%%%%%%%%%
%%%%%%%%%

\section{Monotonicity: numerical results}
\label{sec:num}

In this Appendix the results of numerical calculations are presented. We performed Monte Carlo simulations in which we drew pairs of quantum states and a quantum channel and checked how the optimal quantum transport cost between the selected states changes after the application the channel.

We examined the monotonicity of the optimal quantum transport for qutrits ($N=3$) and ququarts ($N=4$) using two different types of tests. In the first one we picked random quantum channels corresponding to dynamical Choi matrices of a fixed rank. In the second one we employed random extremal quantum channels, defined up to unitary pre- and post-processing of the quantum states. Because of the unitary invariance of $T^Q$ and its convexity (recall Prop. \ref{lem_comb2}), the monotonicity of $T^Q$ under such channels implies the monotonicity with respect to all CPTP maps.

To generate a random quantum state of order $N$ we first drew a rectangular complex random matrix $X$ of size $N \times k$, with $k \leq N$.
%
% with values drawn from a symmetric uniform distribution $U(-1,1)$.
%Number $m$ correspond to rank of the state $\rho$ and was set $m = n$. Then we 
Setting $\rho = X X^{\dagger}/\Tr[X X^{\dagger}]$ we ensure that a random matrix $\rho$ of size $N$ and rank $k$ is non-negative and $\Tr \rho = 1$. Our calculations show that the rank of states does not affect significantly the results. Therefore we focus on full rank states to maximize the explored space. 

Random quantum channels $\Phi: \Omega_N \to \Omega_N$ of a given rank $k$ were generate via the Choi--Jamio{\l}kowski isomorphism. We generated a random $N^2 \times k$ complex matrix $X$ %of complex random  variables form a symmetric uniform distribution 
and created an auxiliary $N^2 \times N^2$ positive matrix $\widetilde{Y} = X X^\dagger \geq 0$. By the Choi--Jamio{\l}kowski isomorphism a quantum state corresponds to a quantum channel if the second of its partial traces is a totally mixed state. To assure this property we defined $Y = \Tr_1 \widetilde{Y}$, where $\Tr_1$ is the partial trace on the first subsystem. Then, we set $\rho_{\Phi} \propto \left(\bone_N \otimes Y^{-1/2}\right) \widetilde{Y} \left(\bone_N \otimes Y^{-1/2}\right)$ with the proportionality coefficient adjusted by the condition $\Tr[\rho_{\Phi}] = 1$ . For further details on random quantum channel generation we encourage the Reader to consult \cite{rand_quant,ZPNC11}.

Having a normalized density matrix $\rho_{\Phi}$, to calculate the action of channel $\Phi$ on any state $\rho$ we used the definition of Choi--Jamio{\l}kowski isomorphism: 

\begin{equation}
\begin{aligned}
& N \Tr_2[\rho_\Phi  ~(\bone_N \otimes \rho^T)] = N \Tr_2\left[(\Phi \otimes \bone)(\vert\psi^+ \ra\la\psi^+\vert ) ~(\bone_N \otimes \rho^T)\right] = \\
& =  \Tr_2\left[(\Phi \otimes \bone_{N\times N})\left(\left(\sum_{i = 1}^N\vert ii \ra \right)\left(\sum_{j = 1}^N \la jj \vert\right) \right) ~(\bone_N \otimes \rho^T)\right] =  \\ 
&= \sum_{ij}\Phi(\vert i \ra\la j\vert)\Tr[\vert i \ra\la j\vert \rho^T] = \Phi(\rho),
\end{aligned}
\end{equation}
where $\vert \psi^+\ra = \frac{1}{\sqrt{N}} \sum_{i = 1}^N \vert ii\ra$ is a maximally entangled state.

In the test of the second type we focused on extremal quantum channels. Using the methods presented in \cite{rand_extr_chann} we described all such channels by their Kraus operators, $K_i = U_i D_i$ with $i = 1, \ldots, N$, where $U_i$ and $D_i$ denote, respectively, the unitary matrices and diagonal matrices of order $N$. An explicit form of $U_i$ and $D_i$
for a qutrit reads, 
\begin{equation*}
\begin{aligned}
& D_1 = \diag(a,b,c) ~,~ D_2 = \diag(d,e,f) ~,~ D_3 = \sqrt{1 - D_1^2 - D_2^2}~, \\
& U_1 = \begin{pmatrix}
1 & 0 & 0\\
0 & 0 & 1\\
0 & 1 & 0
\end{pmatrix} ,~~ 
U_2 = \begin{pmatrix}
0 & 0 & 1\\
0 & 1 & 0\\
1 & 0 & 0
\end{pmatrix} ,~~ 
U_3 = \begin{pmatrix}
0 & 1 & 0\\
1 & 0 & 0\\
0 & 0 & 1
\end{pmatrix} .
\end{aligned}
\end{equation*}
An analogous parametrisation works for $N =4$ (see \cite{rand_extr_chann}). 
All of the parameters in the above formulae are non-negative and chosen so that the matrices $D_3$ (for the qutrit case) and $D_4$ (for the ququart case) are positive semi-definite.

In each Monte Carlo simulation we generated a random pair of initial states $\rho^A, \rho^B$ of order $N$ and a quantum channel $\Phi$ using the method presented above. Then, we calculated the quantum transport cost before and after the application of the channel, $T^Q(\rho^A, \rho^B)$ and $T^Q \big(\Phi(\rho^A), \Phi(\rho^B) \big)$, respectively and compared the results. 

The simulations were programmed in the Python language, using Numpy library for algebraic calculations, Cvxpy library for the transport costs calculations and a solver supplied by the MOSEK Optimizer for solving semidefinite programming problems. The plots were generated with help of the Matplotlib library. The numerical precision of calculations was set to $10^{-12}$. We checked the accuracy of the optimisation algorithm by generating random unitary channels, which correspond to $k=1$ in the procedure described above. The deviation between the obtained values of $T^Q$ before and after the application of a unitary channel was not greater than $10^{-9}$, hence we can take this value as the numerical accuracy of the algorithm.

The results of the simulations are presented in Table \ref{tab:num} and Figure \ref{fig:num}. In all of the considered cases the difference $T^Q(\rho^A, \rho^B) - T^Q(\Phi(\rho^A),\Phi(\rho^B))$ was positive and a few orders of magnitude larger than the numerical accuracy. 
These results allow us to conjecture that optimal quantum transport cost $T^Q$ for dimensions $3$ and $4$ is monotonous under all CPTP maps.

\begin{table}[h]
\begin{center}
\begin{tabular}{l|ll|ll}
& \multicolumn{2}{l|}{$N = 3$, qutrit} & \multicolumn{2}{l}{$N = 4$, ququart} \\ \specialrule{.1em}{.05em}{.05em} 
\textbf{Random channels} & rank $k$ & samples & rank $k$  & samples \\
& 2  & 300 000 & 2   & 200 000  \\
& 3  & 300 000 & 3   & 200 000  \\
& 9  & 600 000 & 4   & 200 000  \\
&    &         & 16  & 400 000   \\ \hline 
$\min \, (T^Q -T^Q\circ\Phi)$  &   & 0.00146 &   & 0.00644 \\ \specialrule{.1em}{.05em}{.05em}  
\textbf{Extremal channels} & samples & 7 500 000  & samples & 7 500 000 \\ \hline
%excluded Mixed unitary  &  samples & & samples &\\ \hline
$\min \, (T^Q -T^Q\circ\Phi)$ &  & 0.00023485  &  & 0.006197901       
\end{tabular}
\caption{\label{tab:num}{\small The results of the Monte Carlo simulations discussed in the text. `Samples' refers to number of drawn triples $(\rho^A, \rho^B, \Phi)$, whereas `$\min (T^Q -T^Q\circ\Phi)$' refers to the smallest difference $T^Q(\rho^A, \rho^B) - T^Q(\Phi(\rho^A),\Phi(\rho^B))$ found in the samples.}}
\end{center}
\end{table}

\begin{figure*}[h!]
    \centering
    \begin{subfigure}[h!]{0.5\textwidth}
        \centering
        \includegraphics[height=2in]{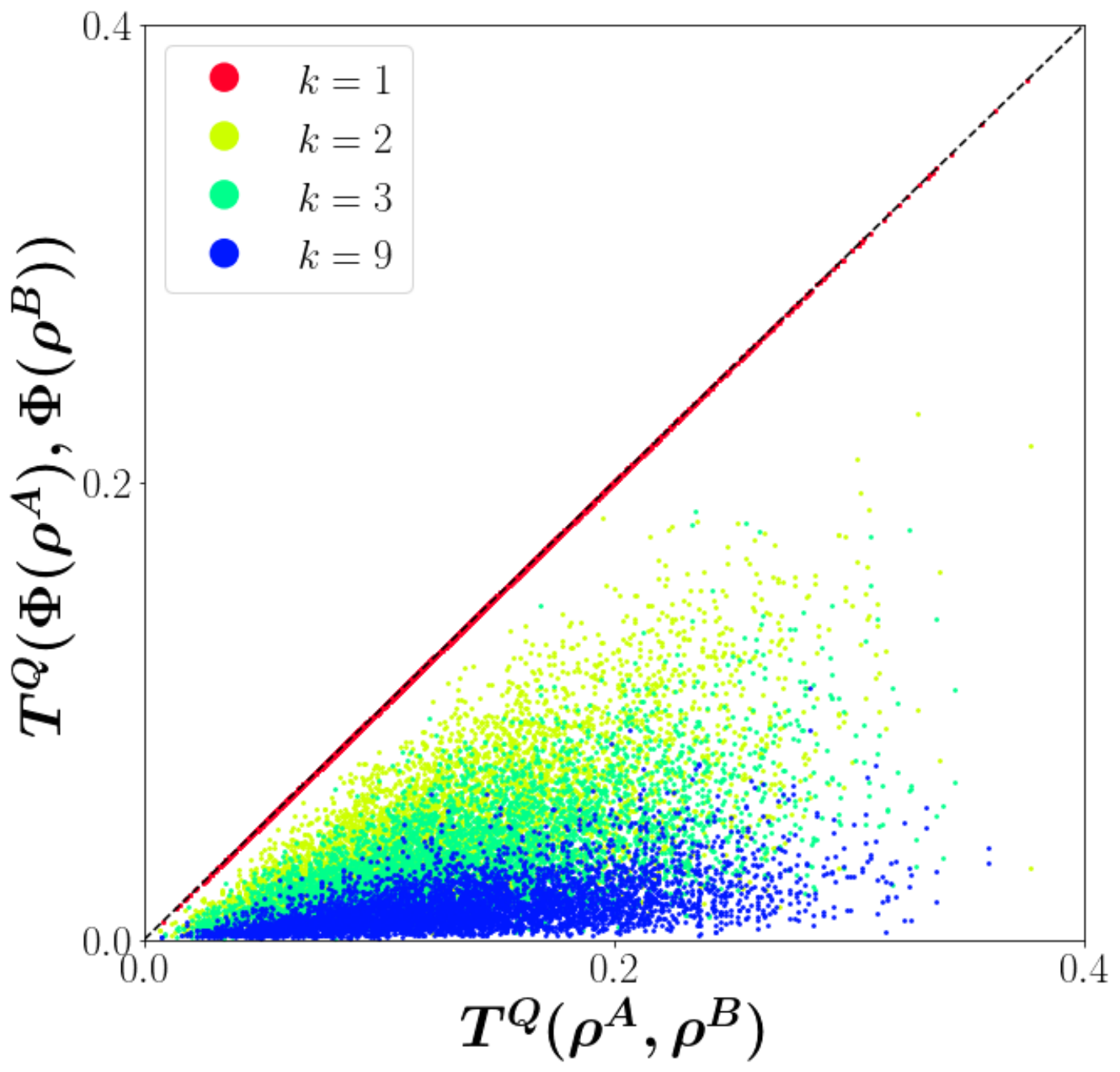}
    \end{subfigure}%
    ~ 
    \begin{subfigure}[h!]{0.5\textwidth}
        \centering
        \includegraphics[height=2in]{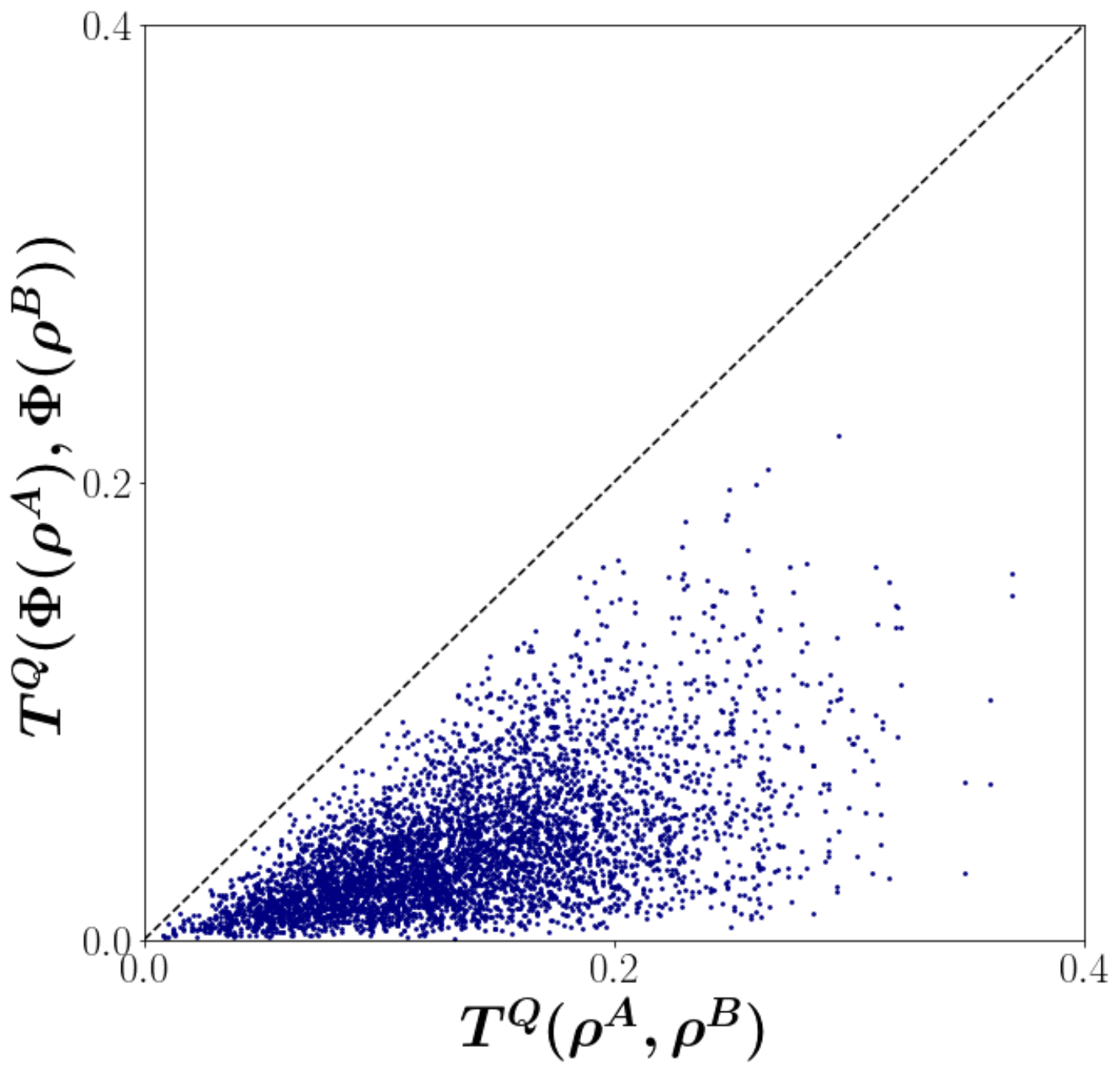}
    \end{subfigure} \\[10 pt]
    \begin{subfigure}[h!]{0.5\textwidth}
        \centering
        \includegraphics[height=2in]{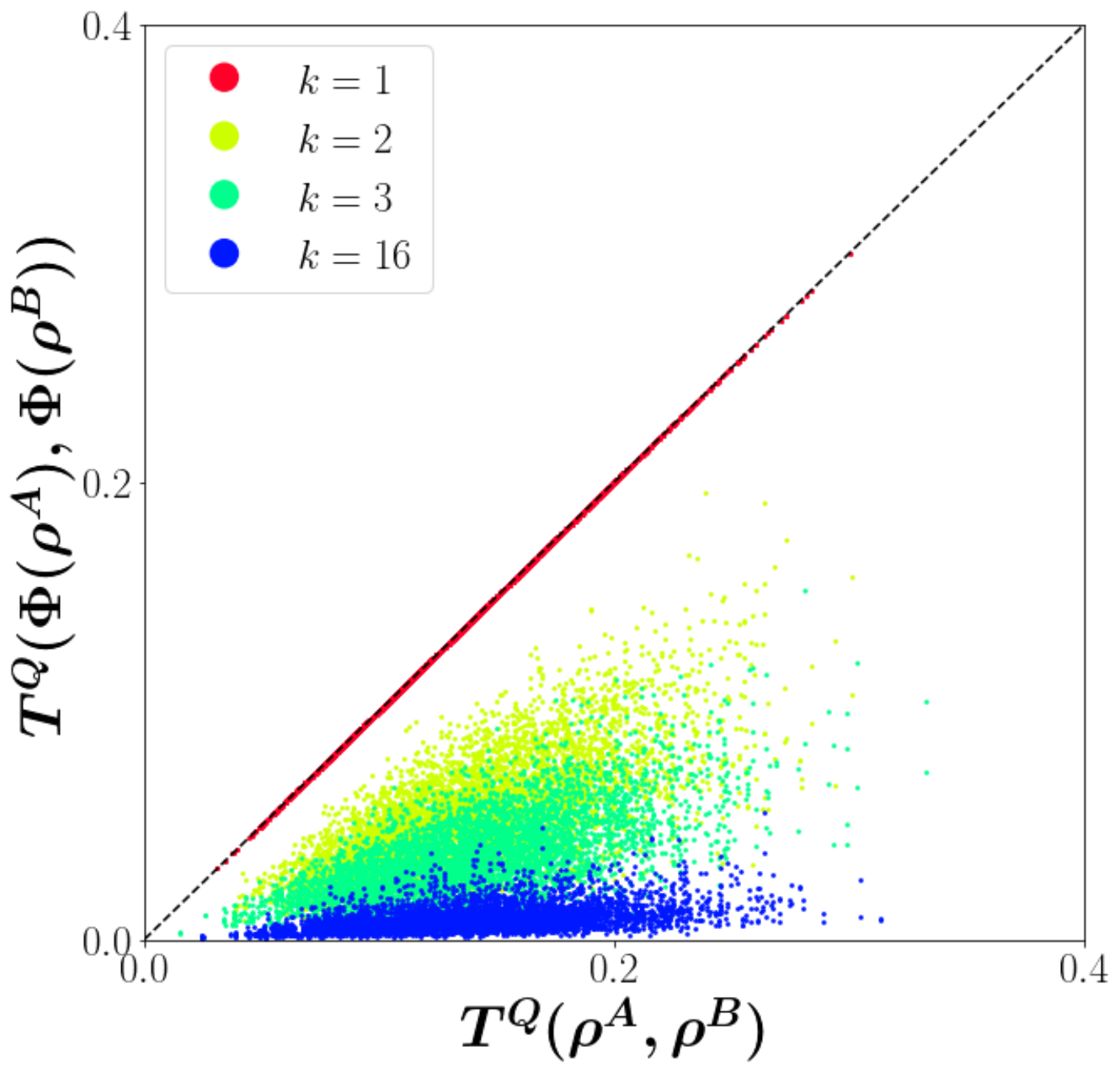}
    \end{subfigure}%
    ~ 
    \begin{subfigure}[h!]{0.5\textwidth}
        \centering
        \includegraphics[height=2in]{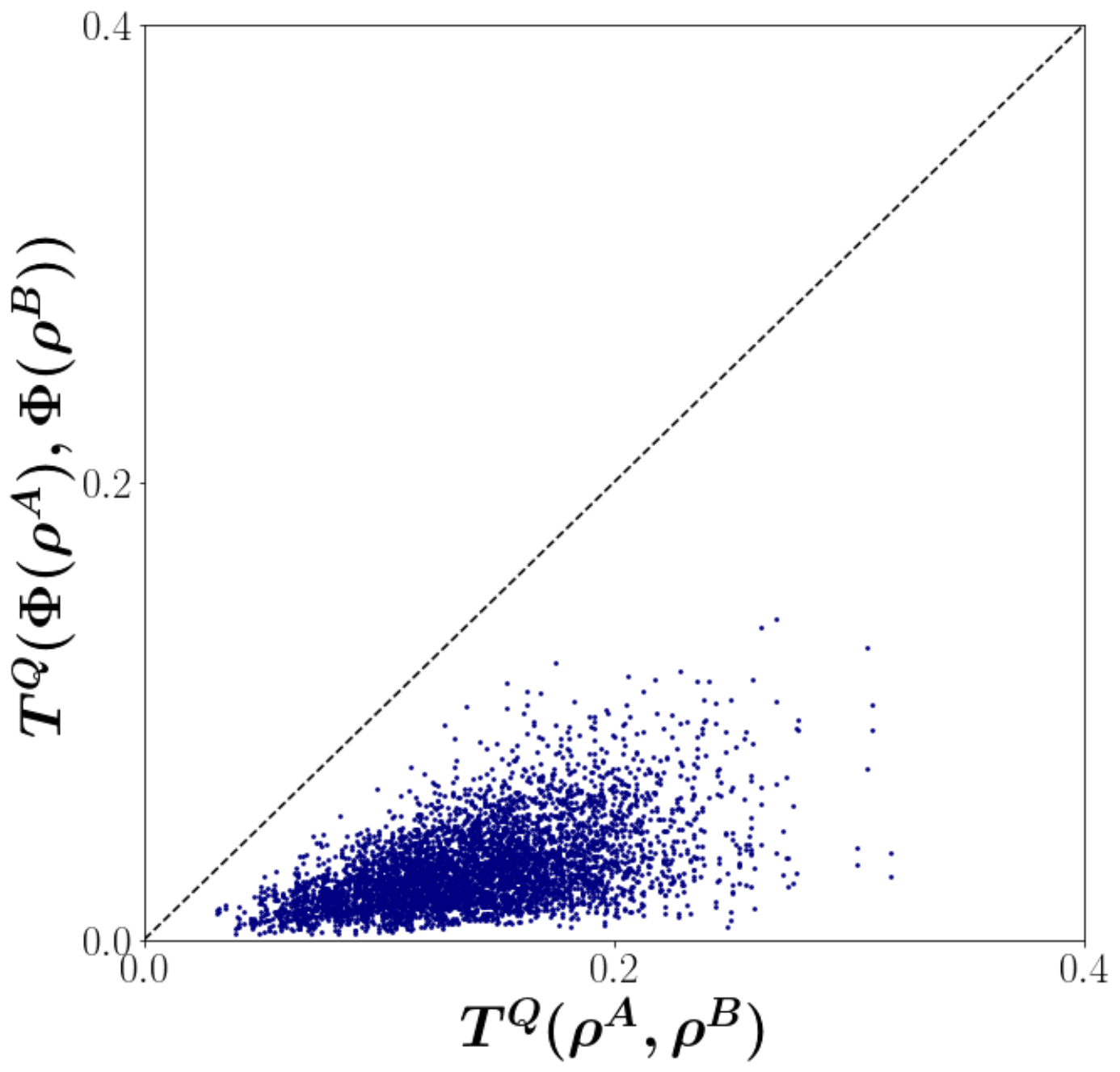}

    \end{subfigure}
    
    \caption{\label{fig:num}\small{The results of the Monte Carlo simulations of random density matrices of order $N=3$ (top) and $N=4$ (bottom) discussed in the text. The optimal quantum transport cost, $T^Q \circ \Phi$, after the action of a quantum channel $\Phi$ is plotted against the initial value of $T^Q$. Results obtained for channels of a fixed rank $k$ are presented on the left and the extremal channels on the right.
    The dashed line marks the equality $T^Q(\rho^A, \rho^B)  = T^Q\big( \Phi(\rho^A), \Phi(\rho^B) \big)$, which is saturated for unitary channels ($k = 1$).}}
\end{figure*}

%    Each sample is interpreted as a triple $(\rho^A, \rho^B, \Phi)$; For each rank $k$ in random channels, and for the extremal channels the  sample size is $5 000$. The horizontal axes correspond to projective quantum transport cost between $\rho^A$ and $\rho^B$. The vertical axes correspond to projective quantum transport cost between $\Phi(\rho^A)$ and $\Phi(\rho^B)$.

\newpage
\newpage

\end{appendices}



\bibliographystyle{plain}

\end{document}